\documentclass[12pt]{article}
\usepackage[utf8]{inputenc}
\usepackage[english]{babel}

\usepackage{times}
\usepackage[margin=1in]{geometry}
\usepackage{amsmath,amsfonts,amsthm,thmtools}
\usepackage{bm}
\usepackage{natbib}

\usepackage{enumitem}
\usepackage{caption}
\usepackage{float}
\usepackage{color}
\usepackage{algorithm,algpseudocode}

\usepackage{indentfirst}
\usepackage{booktabs}
\usepackage[flushleft]{threeparttable}
\usepackage{graphicx}
\usepackage{algorithm}
\usepackage{setspace}

\theoremstyle{plain}
\newtheorem{theorem}{Theorem}
\newtheorem{lemma}{Lemma}
\newtheorem{corollary}{Corollary}

\theoremstyle{remark}
\newtheorem{remark}{Remark}

\declaretheoremstyle[
  bodyfont=\normalfont\itshape,
  headformat=\NAME\NUMBER  
]{nospacetheorem}
\declaretheorem[style=nospacetheorem,name=C]{condition}

\usepackage[section]{placeins}
\usepackage{setspace}
\usepackage{times}

\def\shalf{\mbox{{\footnotesize$\frac{1}{2}$}}}
\DeclareMathOperator{\Var}{Var}
\DeclareMathOperator{\E}{E}
\DeclareMathOperator{\IG}{IG}
\newcommand{\sM}{{\cal M}}
\newcommand{\sN}{{\cal N}}

\newcommand{\sg}{\textnormal{\textsl{g}}}

\newcommand{\pcite}[1]{\citeauthor{#1}'s \citeyearpar{#1}}
\allowdisplaybreaks

\title{\vspace{-0.8in}Model Based Screening Embedded Bayesian Variable Selection for Ultra-high Dimensional Settings}
\author{Dongjin Li, Somak Dutta and Vivekananda Roy\thanks{Corresponding address: Department of Statistics, 2438 Osborn Dr, Ames, IA, 50011. Email: vroy@iastate.edu}\\Department of Statistics, Iowa State University, Ames, IA 50010.}
\date{}

\begin{document}

\maketitle
\begin{abstract}
  We develop a Bayesian variable selection method, called SVEN, based
  on a hierarchical Gaussian linear model with priors placed on the
  regression coefficients as well as on the model space. Sparsity is
  achieved by using degenerate {\it spike} priors on inactive
  variables, whereas Gaussian {\it slab} priors are placed on the
  coefficients for the important predictors making the posterior
  probability of a model available in explicit form (up to a
  normalizing constant). 
  The strong model selection consistency is shown to
  be attained when the number of predictors grows nearly exponentially with
  the sample size and even when the norm of mean effects solely due to the
  unimportant variables diverge, which is a novel attractive feature. 
  An appealing
  byproduct of SVEN is the construction of novel model weight adjusted
  prediction intervals. Embedding a unique model based screening and
  using fast Cholesky updates, SVEN produces a highly scalable
  computational framework to explore gigantic model spaces, rapidly identify the
  regions of high posterior probabilities and make fast inference and prediction. A temperature schedule
  guided by our model selection consistency derivations is used to
  further mitigate multimodal posterior distributions. The performance
  of SVEN is demonstrated through a number of simulation
  experiments and a real data example from a genome wide association
  study with over half a million markers.
\end{abstract}


\section{Introduction}
In almost every scientific discipline, rapid collection of
sophisticated data has been booming due to recent advancements in
technology. In biology, for example, automated sequencing tools have
made whole genome sequencing possible in a cost effective manner, thus
providing variations of millions of single nucleotides between
individuals. On the other hand, because phenotypic data are typically
collected via carefully conducted scientific experiments or other
observational studies, number of observations remains on the smaller
size, giving rise to regression problems where the number of variables
$p$ far exceeds the sample size $n$. Nevertheless, only a few of these
variables are believed to be associated with the response. Thus,
variable selection plays a crucial role in the modern scientific
discoveries.

Classical approaches to deal with the variable selection problems are
through regularization methods. A variety of methods using different
penalization techniques have been proposed for variable selection in
the linear models, such as the lasso \citep{tibs:1996,datt:zou:2017}, SCAD
\citep{fan:li:2001, kim:choi:oh:2008}, elastic net \citep{zou:hast:2005}, adaptive lasso
\citep{zou:2006}, the octagonal shrinkage and clustering algorithm for
regression \citep{bond:reic:2008}, L0-penalty for best subset
regression \citep{bert:etal:2016,huan:etal:2018} and others. These
methods achieve sparsity by either penalizing the effect sizes or the
model sizes but rarely both. Several Bayesian variable selection
methods exploit the connection between the penalized estimators and
the modes of Bayesian posterior densities under suitably chosen prior
distributions on the regression coefficients. Example includes the
lasso-Laplace prior connection \citep{tibs:1996}, the hierarchical
Bayesian lasso \citep{park:case:2008} and other works by
\citet{kyun:gill:ghos:case:2010}, \citet{xu:ghos:2015} and
\citet{roy:chak:2017}.

Another popular approach to Bayesian variable selection is integrating
the penalties on the effect size and the model size via priors
distributions. To that end, auxiliary indicator variables indicating
the presence or absence of each variable are introduced to obtain a
`spike and slab' prior on the regression coefficients. Here the
`spike' corresponds to the probability mass concentrated at zero or
around zero for the variables vulnerable to deletion and the `slab'
specifies prior uncertainty for coefficients of other
variables. Analysis using such models determines (selects) the most
promising variables by summarizing the posterior density of the
indicator variables and/or the regression coefficients. The seminal
works of \citet{mitc:beau:1988,geor:mccu:1993, geor:mccu:1997}
developed a hierarchy of priors over the regression coefficients and
the latent indicators and used Gibbs sampler to identify promising
models in low dimensional setup \citep[see also][]{yuan:lin:2005,
  ishw:rao:2005, lian:paul:moli:clyd:berg:2008,
  john:ross:2012}. Several of these methods have been recently
modified and extended to the ultra-high dimensional
setup. \citet{nari:he:2014} pioneered the theoretical study of
Bayesian variable selection in the ultra-high dimensional setup,
\citet{rock:geor:2014} introduced the EM algorithm for fast
exploration of high-posterior models, \citet{yang:wain:jord:2016}
studied model selection consistency and computational complexity when
$g$-prior is placed on the regression coefficients,
\citet{shin:bhat:2018} extended the popular non-local priors to model
selection and modified the stochastic shotgun model search algorithm
\citep{hans:dobr:2007}, while \citet{zhou:guan:2019} and
\citet{zane:robe:2019} implemented Metropolis Hastings algorithms with
an iterative complex factorization and a tempered Gibbs sampler,
respectively, for estimating posterior model probabilities.


From a practical standpoint, in
the ultra-high dimensional set up, where the number of variables ($p$)
is much larger than the sample size ($n$), generally variable screening is performed to
reduce the number of variables before applying any of the
aforementioned variable selection methods for choosing important
variables. The classical approaches as well as \citet{nari:he:2014}
resort to a two stage procedure where they first use
frequentist screening algorithms \citep{fan:lv:2008,wang:leng:2016} to
reduce the dimension of the problem and then perform variable
selection. \citet{shin:bhat:2018} as well as \citet{cao:khar:2020}
fuse the frequentist iterated sure independent screening in their
stochastic search algorithm. However, these screening methods are
frequentist procedures that are not guaranteed to be fidelitous to the
Bayesian model in practice.

In this work, we extend the classical variable selection model of
\cite{mitc:beau:1988} to the ultra-high dimensional setting. Following
the path laid by \cite{nari:he:2014} we derive posterior consistency
results. By considering zero (exact spike) inflated mixture priors for
regression coefficients, we are able to introduce sparsity and relax
some assumptions of \cite{nari:he:2014}. Furthermore, we develop a
novel methodology for variable selection in the spirit of the
stochastic shotgun search algorithm \citep{hans:dobr:2007} with
embedded screening that is faithful to the hierarchical Bayesian
model. We develop sophisticated computational framework that allows us
to consider larger search neighborhoods and compute exact unnormalized
posterior probabilities in contrast to
\citet{shin:bhat:2018}. Furthermore, in order to recover models with
large posterior probabilities and mitigate posterior multimodality
associated with variable selection models, we use a temperature
schedule that is guided by our posterior model selection consistency
asymptotics. We call this Bayesian method and the computational
framework $s$election of $v$ariables with $e$mbedded scree$n$ing
(SVEN). Keeping prediction of future observations in mind, we develop
novel methods for computing approximate posterior predictive
distribution and prediction intervals. In particular, using SVEN we
construct two prediction intervals, called Z-prediction intervals and
Monte Carlo prediction intervals.


The rest of the paper is laid out as follows. In Section
\ref{sec:sven} we describe the hierarchical Bayesian variable
selection model and prove strong model selection consistency results
(Section \ref{sec:modelAndTheory}); develop the SVEN framework
(Section \ref{sec:methodology}) and prediction methods (Section
\ref{sec:predicion}). We perform detailed simulation studies in
Section \ref{sec:simulation} and compare our methods to several other
popular Bayesian and frequentist methods. In Section
\ref{sec:dataanalysis} we analyze a massive dataset from an
agricultural experiment with $n=3,951$ and $p=546,034$ where among the Bayesian methods used for comparison only our
method is able to perform variable selection on the whole data. We
also show the practical usefulness of our method in obtaining
posterior predictive distribution and prediction intervals for the yield of novel
crop varieties.  We conclude in Section \ref{sec:conclusion} with some
discussion and future research directions. A supplement document
containing the proofs of the theoretical results and some computational details is available with
sections referenced here with the prefix `S'. The methodology
proposed here is implemented in an accompanying R package `bravo' for
$B$ayesian sc$r$eening $a$nd $v$ariable selecti$o$n.

\section{Bayesian variable selection with screening}\label{sec:sven}

\subsection{Hierarchical mixture models} \label{sec:modelAndTheory}
\subsubsection{Model description}
\label{sec:modeldes}
Let $y=(y_1, \ldots, y_n)$ denote a $n \times 1$ vector of response
values, $Z = (Z_1, \ldots, Z_p)$ an $n \times p$ design matrix of $p$
potential predictors, with vector of partial regression coefficients
$\mu \equiv (\mu_1,\ldots,\mu_p).$ We assume latent indicator vector
$\gamma = (\gamma_1,\ldots,\gamma_p) \in \{0,1\}^p$ to denote a model
such that the $j$th predictor is included in the regression model if
and and only if $\gamma_j = 1$. Corresponding to the binary vector,
the size of a model $\gamma$ is denoted as $|\gamma|$, where
$|\gamma|=\sum_{j=1}^p \gamma_j$. Also, with model $\gamma$, let
$Z_{\gamma}$ be the $n \times |\gamma|$ sub-matrix of $Z$ that
consists of columns of $Z$ corresponding to model $\gamma$ and
$\mu_{\gamma}$ be the vector that contains the regression coefficients
for model $\gamma$. In the first hierarchy of the Bayesian
hierarchical mixture model we assume that the conditional distribution
of $y$ given $Z,\gamma,\mu_0,\mu$ and $\sigma^2$ is $n$-dimensional
Gaussian and is given by
\begin{equation}\label{eq:regModelOriginalScale}
 y|Z,\gamma,\mu_0,\mu,\sigma^2 \sim \sN_n(\mu_0 1_n+Z_\gamma\mu_\gamma,\sigma^2I_n),
\end{equation}
where $\mu_0$ is the intercept term and $\sigma^2 > 0$ is the
conditional variance. Thus \eqref{eq:regModelOriginalScale} indicates
that each $\gamma$ corresponds to a Gaussian linear regression model
$y = \mu_01 + Z_\gamma\mu_\gamma + \epsilon$ where the residual vector
$\epsilon \sim \sN_n(0,\sigma^2 I_n).$ However, because the original covariates
could have unbalanced scales, a common approach is to reparameterize
the above model using a scaled covariate matrix. To that end, suppose
$\bar{Z}$ is the vector of column means of $Z$ and $D$ is the
$p\times p$ diagonal matrix whose $i$th diagonal entry is the sample
standard deviation of $Z_i$ (the $i$th column of $Z$) and let
$X = (Z - 1_{n}\bar{Z}^\top)D^{-1}$ denote the scaled covariate
matrix. Also we assume that $\beta = D\mu$ and
$\beta_0 = \mu_0 + \bar{Z}^\top \mu.$ The Bayesian hierarchical
regression model after reparameterization is given by
\begin{subequations}\label{eq:litdiff}
\begin{align}
    y |\beta,\beta_0,\sigma^2, \gamma &\sim \mathcal{N}_{n}\left(1_{n} \beta_{0}+X_{\gamma} \beta_{\gamma}, \sigma^{2} I\right), \label{subeq:regModel}\\
  \beta_{j} | \beta_0, \sigma^2, \gamma &\stackrel{\text { ind }} \sim \mathcal{N}\left(0, \frac{\gamma_{j}}{\lambda} \sigma^{2}\right) \text { for } j=1, \ldots, p,\label{subeq:priorBeta}\\
   \left(\beta_{0}, \sigma^2\right)|\gamma &\sim f\left(\beta_{0}, \sigma^{2}\right) \propto 1 / \sigma^{2}, \label{subeq:priorInterceptVariance}\\
    \gamma|w &\sim f(\gamma | w)=w^{|\gamma|}(1-w)^{p-|\gamma|}\label{subeq:priorGamma}.
\end{align}
\end{subequations}
In this hierarchical setup a popular non-informative prior is set for
$(\beta_0, \sigma^2)$ in \eqref{subeq:priorInterceptVariance} and a
conjugate independent normal prior is used on $\beta$ given $\gamma$
in \eqref{subeq:priorBeta} with $\lambda > 0$ controlling the
precision of the prior independently from the scales of
measurements. Note that under this prior, if a covariate is not
included in the model, the prior on the corresponding regression
coefficient \emph{degenerates at zero.} In \eqref{subeq:priorGamma} an
independent Bernoulli prior is set for $\gamma$, where $w \in (0, 1)$
reflects the prior inclusion probability of each predictor. We assume $\lambda$ and $w$ are known non-random functions of
$n$ and $p.$

The hierarchical model \eqref{eq:litdiff} with centered $X$ allows us to obtain the
distribution of $y$ given $\gamma$ in a closed form by integrating out
$\beta_0,$ $\beta_\gamma$ and $\sigma^2$ \citep[][section S6]{roy:tan:fleg:2018}. Consequently, the marginal
likelihood function of $\gamma$ is given by
\begin{align}\label{eq:marginalGamma}
    L(\gamma|y) &= \int_{\mathbb{R}_{+}}  \int_{\mathbb{R}^{\gamma}} \int_{\mathbb{R}} f\left(y | \gamma, \sigma^{2}, \beta_{0}, \beta_{\gamma}\right) f\left(\beta_{\gamma} | \gamma, \sigma^{2}, \beta_{0}\right) f\left(\sigma^{2}, \beta_{0}\right)  d \beta_{0} d \beta_{\gamma} d \sigma^{2} \nonumber \\
    & = c_{n,p} ~\lambda^{|\gamma|/2} |A_\gamma|^{-1/2} R_\gamma^{-(n-1)/2},
\end{align}
where $A_\gamma = X_{\gamma}^{\top} X_{\gamma}+\lambda I,$ $|A_\gamma|$ is the determinant of $A_\gamma,$
\begin{equation}
  \label{eq:rgam}
  R_{\gamma} = \tilde{y}^{\top} \tilde{y}-\tilde{y}^{\top}X_{\gamma}A_\gamma^{-1}X_\gamma^\top\tilde{y} =  \tilde{y}^{\top} \tilde{y}-\tilde{\beta}_{\gamma}^{\top}A_\gamma \tilde{\beta}_{\gamma} = \tilde{y}\left( I + \lambda^{-1}X_\gamma X_\gamma^\top \right)^{-1}\tilde{y}  
\end{equation}
 is the ridge residual sum of squares, $\tilde{y}=y-\bar{y} 1_{n}$, $\bar{y} = \sum_{i=1}^n y_i/n,$ $\tilde{\beta}_{\gamma}=A_\gamma^{-1} X_{\gamma}^{\top} \tilde{y}$ and $c_n = \Gamma((n-1) / 2) / \pi^{(n-1) / 2}$ is the 
normalizing constant. 

In order to identify the important variables, we use the (marginal) posterior
distribution of $\gamma$. Thanks to the explicit form of the marginal
likelihood \eqref{eq:marginalGamma}, this posterior density is
given by
\begin{equation*}
    f(\gamma | y) \propto f(y | \gamma) f(\gamma) \propto \lambda^{|\gamma|/2} |A_\gamma|^{-1/2} R_\gamma^{-(n-1)/2} w^{|\gamma|}(1-w)^{p-|\gamma|}.
\end{equation*}
It's often convenient to work with log of the posterior density which is given by
\begin{equation}\label{eq:logpostGamma}
\log f(\gamma|y) = const + \shalf|\gamma|\log\lambda - \shalf\log|A_\gamma| - \shalf(n-1)\log R_\gamma + |\gamma|\log(w/(1-w)).
\end{equation}

\begin{remark}
  It is important to note that the regression model
  \eqref{eq:regModelOriginalScale} on the original covariate scale
  should be used for prediction, instead of \eqref{eq:litdiff} because
  the hierarchical prior \eqref{subeq:priorBeta} is defined under the
  assumption that $1_{n}^\top X = 0$ and $X_j^\top X_j = n$ for all
  $j.$
\end{remark}
\begin{remark}
  In this work we assume $w$ is fixed. However, a popular alternative
  is to assign a Beta prior on $w$, i.e., let
  $w \sim f(w) \propto w^{a-1}(1-w)^{b-1}$ for some $a, b >0$. Then it is possible to
  integrate out $w$ from \eqref{subeq:priorGamma} to obtain the
  marginal prior distribution of $\gamma$ given by
  $f(\gamma)=B(|\gamma|+a, p-|\gamma|+b)/{B(a, b)},$ where
  $B(\cdot, \cdot)$ is the beta function. This will replace the last
  term in \eqref{eq:logpostGamma} by $\log f(\gamma).$
\end{remark}

\begin{remark}
As an alternative to the independent normal prior \eqref{subeq:priorBeta}, it is also possible to consider Zellner's $g$-prior \citep{zell:1986} on $\beta_\gamma$ given by $\beta_\gamma|\gamma,\sigma^2 \sim$ 
 $\sN_{|\gamma|}\left( 0, g\sigma^2 (X_\gamma^\top X_\gamma)^{-1}\right)$
 provided that for every $k \leq n-1,$ all $n\times k$ submatrices of $X$ have full column rank and we restrict the support of the prior distribution on $\gamma$ to models of size at most $n-1.$ Assuming that $g$ is a non-random function of $n$ and $p,$ the marginal posterior of $\gamma$ is then given by
\[f_g(\gamma|y) \propto  \left[\tilde{y}^\top\tilde{y} - \frac{g}{g+1}\tilde{y}^\top X_\gamma (X_\gamma^\top X_\gamma)^{-1}X_\gamma^\top \tilde{y}\right]^{-(n-1)/2}
\frac{w^{|\gamma|}(1-w)^{p-|\gamma|}}{(1+g)^{|\gamma|/2}}\mathbb{I}(|\gamma| < n),\]
where the priors on $\beta_0$ and $\sigma^2$ have been assumed to be the same as \eqref{subeq:priorInterceptVariance}.
 \end{remark}

 Ideally, as the sample size increases we would like the posterior of
 $\gamma$ to concentrate more and more on the important
 variables. Several works have alluded to asymptotic guarantees for
 strong model selection consistency in the ultra-high dimensional
 regression where $p$ is allowed to vary subexponentially with $n,$
 i.e. $\min\{n,p\} \to \infty$ and $(\log p)/n\to 0.$ Here, the strong
 model selection consistency implies that the posterior probability of
 the \emph{true set} of variables converge to 1 as $n$ tends to
 infinity.  Under shrinking and diffusing priors, \citet{nari:he:2014}
 developed explicit scaling laws for hyper-parameters that are
 sufficient for strong model selection consistency. On the other hand,
 \citet{shin:bhat:2018}, and more recently, \citet{cao:khar:2020}
 established sufficient conditions for strong model selection
 consistency under non-local type priors \citep{john:ross:2012}. The
 Bayesian hierarchical model \eqref{eq:litdiff} is similar to
 \pcite{nari:he:2014} model with the crucial distinction that the
 spike prior is degenerate: $P(\beta_i = 0|\gamma_i = 0) = 1.$
 Consequently, although most assumptions used here for selection consistency are similar to those made by
 \citet{nari:he:2014}, we are able to relax some of the conditions to allow for
 more noisy unimportant variables. In the next section we describe
 strong model selection consistency results for \eqref{eq:litdiff}. 

\subsubsection{Model selection consistency}
We consider the ultra-high dimensional setting where the number of
variables $p$ is allowed to vary subexponentially with the sample
size. As established by \citet{nari:he:2014} the slab precision
$\lambda$ also needs to vary with $n$ for strong model selection
consistency.  In order to state the assumptions and the main
results, we use the following notations. Abusing notation, we
interchangeably use a model $\gamma$ either as a $p$-dimensional
binary vector or as a set of indices of non-zero entries of the binary
vector. For models $\gamma$ and $s,$ $\gamma^c$ denotes the complement
of the model $\gamma$, and $\gamma\vee s$ and $\gamma\wedge s$ denote
the union and intersection of $\gamma$ and $s$, respectively. For two
real sequences $(a_n)$ and $(b_n)$, $a_n \sim b_n$ means
${a_n}/{b_n}\rightarrow c$ for some constant $c>0$; $a_n \succeq b_n$
(or $b_n \preceq a_n$) means $b_n = O(a_n)$; $a_n \succ b_n$ (or
$b_n \prec a_n$) means $b_n = o(a_n)$. Also for any matrix $A$, let
$\alpha_{min} (A)$ and $\alpha_{max}(A)$ denote its minimum and
maximum eigenvalues, respectively, and let $\alpha_{min}^{*}(A)$ be
its minimum nonzero eigenvalue. Again, abusing notations, for two real
numbers $a$ and $b$, $a \vee b$ and $a \wedge b$ denote max$(a, b)$ and min$(a, b),$ respectively. Define
$r_{\gamma}=\text{rank}(X_{\gamma})$ and for $\nu>0$,
$r_{\gamma}^{*}=r_{\gamma} \wedge u_n(\nu)$ where
$$u_n(\nu)=p \wedge \frac{n}{(2+\nu)\text{log}{p}} \quad\textrm{ and }\quad \eta_m^n(\nu) = \underset{|{\gamma}| \leq u_n(\nu)}{\inf} {\alpha_{min}^{*}(X^{\top}_{\gamma} X_{\gamma}/n)}.$$
Finally for any fixed positive integer $J$, define 
\[
		\Delta_n(J)=\underset{\{\gamma :|\gamma|<J|t|, \gamma \not\supset t\}}{\inf} \Vert(I-P_{\gamma})X_{t}\bm{\beta}_{t} \Vert^2,
\]
where  $P_\gamma =  X_\gamma(X_\gamma^\top X_\gamma)^{-}X_\gamma^\top$
is  the  orthogonal  projection  matrix   onto  the  column  space  of
$X_\gamma$  and $\Vert  \cdot  \Vert$ denotes  the  $L_2$ norm. Here,
$A^{-}$  denotes  the Moore-Penrose  inverse  of  $A$. We  assume  the
following set of conditions.

\begin{condition}\label{c.dim}
    $p=e^{nd_n}$ for some $d_n \rightarrow 0$ as $n \rightarrow \infty$, that is, $\text{log}{p} = o(n)$. 
\end{condition}
\begin{condition}\label{c.shrink_rates}
    $n/\lambda \sim (n \vee  p^{2+3\delta})$ for some $\delta>0$, and $w \sim p^{-1}$.
\end{condition}
\begin{condition}\label{c.tc_bound}
    $y = \beta_0 1_n + X_t\beta_t + X_{t^c}\beta_{t^c} + \epsilon$ where $\epsilon\sim\sN(0,\sigma^2I_n),$ the true model $t$ is fixed and $\Vert X_{t^c}\beta_{t^c} \Vert \preceq \sqrt{\log p}$.
\end{condition}
\begin{condition}\label{c.eigen_g}
    For $\delta$ given in C\ref{c.shrink_rates}, there exists $J>1+8/\delta$ such that $\Delta_n(J) \succ \log(\sqrt{n} \vee p)$, and for some $\nu<\delta$, $\kappa<(J-1)\delta/2,$\mbox{}\\
	\centerline{$\eta_m^n(\nu) \succeq \left(\frac{n \vee p^{2+2\delta}}{n/\lambda} \vee p^{-\kappa} \right)$.} 
\end{condition}
\begin{condition}\label{c.eigen_t}
    For some positive constants $a_0$ and $b_0,$ $a_0 < \alpha_{min}\left(\frac{X_t^{\top} X_t}{n}\right) < \alpha_{max}\left(\frac{X_t^{\top} X_t}{n}\right) < b_0$  $\forall n.$
\end{condition}

The condition C\ref{c.shrink_rates} states that the conditional
distribution of $\beta_i$ given $\gamma_i = 1$ is diffused in the
sense that it's conditional prior variance goes to infinity at a
particular rate. The condition C\ref{c.tc_bound} greatly relaxes the
boundedness assumption on $\|X_{t^c}\beta_{t^c}\|$ in
\citet{nari:he:2014}, by slightly strengthening the identifiability
condition C\ref{c.eigen_g}. \cite{yang:wain:jord:2016} obtained similar results under $g$-priors on $\beta$ but as mentioned by them our independence prior is `a more realistic choice'. Moreover, \cite{yang:wain:jord:2016} assumed that $\alpha_{\textrm{min}}(X_\gamma^\top X_\gamma/n)$ is bounded away from zero for all models $\gamma$ of size at most $O(n/\log p),$ which is unrealistic because, for example, even when entries of $X$ are iid N(0,1), $\inf_{1\leq i \leq p}X_i^\top X_i/n$ converges to zero in probability.
Because of the degenerated form of the
spike priors, the regularity assumptions on the submatrices of the
design matrix $X$ in C\ref{c.eigen_g} relax the assumptions on the
bound on their largest eigenvalues. \cite{nari:he:2014} showed that if
the rows of $X$ are independent isotropic sub-Gaussian random vectors
then C\ref{c.eigen_g} holds with overwhelmingly large probability \cite[see also][]{chen:chen:2008,
  kim:kwon:choi:2012, shin:bhat:2018}. The regularity assumption for
the true model C\ref{c.eigen_t} is standard and has been used in both
\citet{nari:he:2014} and \citet{cao:khar:2020} without being
explicitly stated.

Note that the condition C\ref{c.tc_bound} does not explicitly specify
the true model $t$ and the relaxation to allow higher noise
$\|X_{t^c}\beta_{t^c}\|$ warrants a validation of the identifiability
of $t$. To that end, suppose on the contrary that it is possible to include some variables,
say $s$ from $t^c$ into the true model and still maintain the conditions C1-C5
for both $t$ and $t\vee s$ as true models for every $n.$ Then condition
C\ref{c.eigen_g} with $\gamma = t$ (now excluding the apparently true
variable $s$) would imply
$\|(I-P_t)X_s\beta_s\|^2 = \|(I-P_t)(X_t\beta_t + X_s\beta_s)\|^2 =
\|(I-P_t)X_{t \vee s} \beta_{t \vee s}\|^2 \succ \log(p\vee \sqrt{n}).$
Here, the first equality follows from the fact that
$P_t X_t= X_t$. But because $I-P_t$ is
symmetric and idempotent,
\begin{equation}\label{eq.contradiction1}
\|X_s\beta_s\| \geq \|(I-P_t)X_s\beta_s\|  \succ \sqrt{\log(p\vee\sqrt{n})} \geq \sqrt{\log p}. 
\end{equation}
 However, condition C\ref{c.tc_bound} for $t\vee s$ implies $\| X_{t^c\wedge s^c}\beta_{t^c\wedge s^c}\| \preceq \sqrt{\log p}.$ This with \eqref{eq.contradiction1} implies that
 \[\|X_{t^c}\beta_{t^c}\| = \|X_s\beta_s + X_{t^c\wedge s^c}\beta_{t^c\wedge s^c}\| \geq \|X_s\beta_s \| - \| X_{t^c\wedge s^c}\beta_{t^c\wedge s^c}\| \succ \sqrt{\log p}, \]
 which contradicts condition C\ref{c.tc_bound}. We now present the strong model selection consistency results. 

\begin{theorem}\label{thm.consistency}
Assume conditions C1--C5 hold and that $\sigma^2$ is known. Then the posterior probability of the true model, $f(t|y,\sigma^2) \to 1$ in probability as the sample size $n$ approaches $\infty$.
\end{theorem}
\begin{proof}
 The proof is given in Section \ref{sec:proofKnownSigma} of the supplementary materials.
\end{proof}

Note that the statement of Theorem \ref{thm.consistency} is equivalent to $\big(1-f(t|y, \sigma^2)\big)/f(t|y, \sigma^2) \to 0$ in probability as $n \to \infty.$ The proof of Theorem \ref{thm.consistency} also provides the rate of convergence given by,
\begin{equation*}
   \frac{1-f(t|y, \sigma^2)}{f(t|y, \sigma^2)} \preceq \exp\{-v n\} + \rho_n + \rho_n^{(J-1)|t|/2} + \exp \left\{-v'\left(\Delta_{n}(J)-v'' \log(\sqrt{n}\vee p) \right)\right\}
\end{equation*}
with probability greater than $1-\big[2\exp\{-cn\} + 2\exp\{-c'\log p\} + \exp\{-c''\Delta_n(J)\}\big]$
for some positive constants $v$, $v'$, $v''$, $c$, $c'$ and $c''$,  where $\rho_n = p^{-\delta / 2} \wedge \left(p^{1+\delta / 2}/\sqrt{n}\right).$ It is encouraging that despite relaxing the boundedness condition on $\Vert X_{t^c} \beta_{t^c} \Vert$, the rate of convergence remains the same as in \citet{nari:he:2014}.

However, in practice $\sigma^2$ is typically never known. In this case, we need a further assumption that assigns a prior probability of zero on $\widetilde{M} = \{\gamma: r_\gamma > r_t + n/[(2+\nu')\log p]\}$ for some $\nu' > \nu \vee (2/\delta)$.
\begin{condition}\label{c.truncateGamma}
For some $\nu >0$ and $\nu' > \nu \vee (2/\delta)$, $P\left(\gamma \in \widetilde{M}\right)=0$.
\end{condition}
This condition is same as in \citet{nari:he:2014} and also equivalent to the assumptions on the prior model sizes in \cite{shin:bhat:2018} and \cite{cao:khar:2020}.

\begin{theorem}\label{thm.consistency.nosigma}
 Assume conditions C1--C6 hold. Then the posterior probability of the true model, $f(t|y) \to 1$ in probability as the sample size $n$ approaches $\infty$.
\end{theorem}
\begin{proof}
 The proof is given in Section \ref{sec:proofUnknownSigma} of the supplementary materials.
\end{proof}

%


Note that strong consistency results also imply that with probability tending to one, the true model is the posterior mode, that is,  $P(t = \arg\max_\gamma f(\gamma|y)) \to 1$ as $n\to\infty.$ However, in finite
sample this need not be true. 
Furthermore, when the regularity
conditions do not hold, there may be multiple models with large posterior probabilities even for large $n$. Thus, we would like to discover
models with practically large posterior
probability values.
However, in ultra-high dimensional problems,
traditional computational methods based on Markov chain Monte Carlo
(MCMC) algorithms have poor performance. Thus next we describe SVEN to
explore the posterior distribution $f(\gamma|y).$ In particular, SVEN will be used
to discover high probability regions and find the maximum a posteriori (MAP) model
$\arg \max_\gamma f(\gamma|y)$.

\subsection{Searching for high posterior
probability models}\label{sec:methodology}
\subsubsection{Stochastic shotgun search algorithms}
\citet{hans:dobr:2007} proposed the stochastic shotgun search (SSS)
algorithm for recovering models with large posterior probabilities. To
that end, for a given model $\gamma$ let
$\text{nbd}(\gamma) = \gamma^{+} \cup \gamma^{\circ} \cup \gamma^{-}$
denote a neighborhood of $\gamma$, where $\gamma^{+}$ is an ``added"
set containing all the models with one of the $p-|\gamma| $ remaining
covariates added to the current model $\gamma$, $\gamma^{-}$is a
``deleted" set obtained by removing one variable from $\gamma;$ and
$\gamma^{\circ}$ is a ``swapped" set containing the models with one of
the variables from $\gamma$ replaced by one variable from $\gamma^c.$
The SSS algorithm then starts with an initial model $\sg^{(0)},$ and
for $k=1,2,\ldots$
\begin{enumerate}
\item[-] (SSS1) Compute $f(\gamma|y)$ for all $\gamma \in$
  nbd($\sg^{(k-1)}$).
 \item[-] (SSS2) Separately sample $s^+$ from $\sg^{(k-1)+},$  $s^\circ$ from $\sg^{(k-1)\circ}$ and $s^-$ from $\sg^{(k-1)-}$ with probabilities proportional to $f(\cdot|y).$
 \item[-] (SSS3) Sample $\sg^{(k)}$ from $s^+,s^\circ$ and $s^-$ with probability proportional to $f(s^+|y),$ $f(s^\circ|y)$ and $f(s^-|y)$ respectively.
 \end{enumerate}
After running for some prespecified large number of iterations, the algorithm then declares the model discovered with the largest (unnormalized) posterior probability as the MAP model. \citet{hans:dobr:2007} notes that the sampling probabilities in (SSS1) and (SSS2) can be replaced by the Bayesian information criteria (BIC) and the sampling weights can be computed in parallel.

Following the success of SSS, \citet{shin:bhat:2018} propose further
improvement. Note that, \citet{shin:bhat:2018} use non-local priors,
and so the posterior probabilities $f(\gamma|y)$ are not available
analytically. In fact, they resort to using computationally expensive
Laplace approximation which suggests exact numerical computations of
these quantities are also not straightforward \citep[see
also][]{cao:khar:2020}. Also in ultra-high dimensional problems, SSS
may not be scalable due to its implementation. Thus
\cite{shin:bhat:2018} propose a simplified stochastic shotgun search
with screening (S5) by dropping the ``swapped'' set from consideration
and moreover, by screening out variables from the ``added'' set. (Note
that, in high dimension, the number of ``swapped'' models is much
larger than the numbers of ``added'' and ``deleted'' models.) For
screening, borrowing ideas from frequentist correlation screening of
\citet{fan:lv:2008}, they propose computing the least squares
residuals from a regression of $y$ on $X_\gamma$ and compute the
absolute correlations between each column of $X_{\gamma^c}$ and
the residuals. They then propose keeping models in the
``added'' set corresponding to the largest few of the absolute
correlations. This greatly reduces the burden of computing
$f(\gamma|y)$ for all $\gamma$ in the ``added'' set. However, in their
R package BayesS5, the authors have used ridge residuals with unit
ridge penalty instead of the least squares residuals. Nevertheless,
the S5 algorithm has been useful for exploring the posterior
distribution of $\gamma$ \citep{cao:khar:2020}.

In the variable selection model \eqref{eq:litdiff}, the Gaussian
conjugacy provides analytically tractable forms for $f(\gamma|y)$ up
to a normalizing constant. We also show that $f(\gamma|y)$ can be
rapidly computed for the swapped models, thereby allowing us to
include the swapped models in the neighborhood. We thus develop a
stochastic shotgun algorithm with (posetrior) model based screening
and develop scalable statistical computations for drawing fast
Bayesian inference and prediction.

\subsubsection{Selection of variables with embedded screening}\label{sec:alg}
In order to describe the SVEN algorithm, we first describe how to
compute the unnormalized posterior probabilities in the (SSS1)
step. To that end, compute $\zeta = X^\top\tilde{y}$ as $D^{-1}Z^\top\tilde{y}$ once and for all. Next, suppose we have a current model $\gamma$ and we
want to compute the posterior probabilities of each model in
$\gamma^+.$ Suppose $U_\gamma$ is the upper triangular Cholesky factor
of $X_\gamma^\top X_\gamma + \lambda I$ and
$v_\gamma = U_\gamma^{-\top}X_\gamma^\top \tilde{y}.$ In the algorithm
below, scalar addition to vector, division between two vectors and
other arithmetical and algebraic operation on vectors are interpreted
as entry-wise operations, as implemented in most statistical software
(e.g. in \textsf{R}). Then
\begin{enumerate}
 \item Compute $S_1 \leftarrow U_{\gamma}^{-\top}X_\gamma^\top$ by using forward substitution.
 \item Update $S_2 \leftarrow S_1ZD^{-1}.$ [No need to center Z because $S_11=0.$]
 \item Compute $S_{3}$ as the sum of squares of each column of $S_2.$ Note that $S_2$ is a $|\gamma|\times p$ matrix and so these sums of squares should be computed without storing another $|\gamma|\times p$ matrix.
 \item Set $S_4 \leftarrow \sqrt{n+\lambda - S_3}$ where the arithmatical operations are performed entrywise on the vector. Also in this operation, the entries corresponding to the variables in $\gamma$ are ignored.
 \item Compute $S_5 \leftarrow (\zeta - S_2^{\top}v_\gamma)/ S_4.$
 \item Compute $S_6 \leftarrow \log\det U_{\gamma} + \log S_4$ 
 
 \item Compute $S_7  \leftarrow \|\tilde{y}\|^2 - \|v_\gamma\|^2 - S_5^2$
 
 \item Compute $S_8 \leftarrow 0.5(|\gamma|+1)\log\lambda - S_6 - 0.5(n-1)\log S_7 + (|\gamma+1|)\log(w/(1-w)).$
 \end{enumerate}
 Then for all $i\notin \gamma,$ the $i$th entry of $S_8$ above
 contains the unnormalized posterior probability of the model obtained
 by including $i$ in $\gamma.$ The other entries are ignored. For each
 model in $\gamma^{-},$ its posterior probability can be computed
 easily because typically $|\gamma|$ is small. Furthermore, for each
 $\gamma'\in\gamma^{-}$ we can use the above algorithm to compute the
 unnormalized posterior probabilities of $\gamma''$ in $\gamma'^{+}.$
 Thus we can compute the (unnormalized) posterior probabilities of each model in
 nbd$(\gamma).$

 Given the current model $\gamma$, the complexity for computing
 (unnormalized) $f(\gamma| y)$ for all $\gamma \in$ nbd$(\gamma)$ by
 the above algorithm is
 $\mathcal{O}(|\gamma|^3n + |\gamma|^4 + |\gamma|^2 \|Z\|_0+
 |\gamma|^2 p + p)$, where $\|Z\|_0$ denotes the number of non-zero
 elements in $Z$. Since $|\gamma|$ is practically finite, the
 computational complexity is simply $\mathcal{O}(n \vee p +
 \|Z\|_0)$. If in addition, $Z$ is sparse, as in the genome-wide
 association study example in section~\ref{sec:dataanalysis}, the
 complexity for computing all posterior probabilities in nbd$(\gamma)$
 is linear in both $n$ and $p$. Finally, note that, the additional memory
 requirement for the above algorithm except storing the $Z$ matrix is practically $\mathcal{O}(n \vee p)$. Also, different steps including step 2 of the above algorithm can be performed in parallel using distributed computing architecture.  

 Using the above algorithm as the foundation, we now discuss the SVEN
 algorithm. Suppose $1 = T_1 < T_2 < \cdots <T_m$ is a given
 temperature schedule. Let $\sg^{(0)}$ denote the empty model (i.e.
 the model without any predictor included). Then, for $i=1,2,\ldots,m$
\begin{itemize}
 \item[-] Set $\sg^{(i,0)}$ to be the empty model. Then for $k=1,\ldots,N$
 \item[-] (SVEN1) [Same as (SSS1)] Compute $f(\sg'|y)$ for all $\sg'\in$ nbd$(\sg^{(i,k-1)}).$
  
 \item[-] (SVEN2) [Screening step] Consider \emph{at most} 20
  highest probability neighboring models. That is, construct the set $\sM_k \subseteq $
   nbd$(\sg^{(i,k-1)})$ with $|\sM_k| \leq 20$ such that
   $\sg' \in \sM_k$ only if
   $$f(\sg'|y) > \varrho \max_{g''\in \textrm{nbd}(g^{(i,k-1)})}f(g''|y)$$
   and
   $f(\sg'|y) \geq f(\sg''|y),$ $\forall \sg'' \in \textrm{nbd}(g^{(i,k-1)})\cap\sM_k^c,$ where
   $\varrho$ is some prespecified number (we use $\varrho=\exp(-6)$).
 \item[-] (SVEN3) [Shotgun step] Assign the weight $f(\sg'|y)^{1/T_i}$ to a model $\sg' \in \sM_k.$ Sample a model from $\sM_k$ using these weights and set it as $\sg^{(i,k)}.$
 \end{itemize}
 Our ability to efficiently compute posterior probability of {\it all}
 neighboring models allows us to implement the screening (SVEN2) {\it
   directly} using the objective function $f(\gamma | y)$. This is a
 key difference between SVEN and S5 of \citet{shin:bhat:2018}. Because
 models with large probabilities could be separated by models with
 very low probabilities, a temperature schedule has been used. Such
 tempering is quite common in simulated annealing
 \citep{kirk:gela:1983} and has also been used in
 \cite{shin:bhat:2018}. In order to choose a temperature schedule, we
 turn to our asymptotic results from Section
 \ref{sec:proofKnownSigma}. In particular, the theory indicates that
 the log-posterior probabilities of good models with small model size
 are separated by roughly $O(\log p).$ Thus in order to facilitate
 jumps between these models we set $T_m = \log p + \log\log p $ where
 the additional $\log\log p$ is a heuristic adjustment common in
 numerical computations.  Also the remaining temperatures are chosen
 to be equally spaced between 1 and $T_m.$

Note that at every temperature we start the SVEN algorithm at the
empty model that are run separately. Because the stochastic shotgun
might have a tendency to wander off to obscure valleys containing
large number of variables especially under high temperature; running
them separately avoids getting trapped in such a valley. Most good
models have small size and so they could be explored relatively early
when started multiple times from the empty model.

Note that our algorithm does not require explicitly storing the matrix
$X.$ Indeed, in many applications, $Z$ could be sparse and efficiently
stored in the memory. The matrix $X$ on the other hand is always
dense. Overall our method is extremely memory efficient, and we are
able to directly perform variable selection with significantly larger
$p$ than the other methods may handle. 
 
In addition to the MAP model, our method also provides the posterior
probability of all the models explored by the algorithm and facilitate
approximate Bayesian model averaging \citep{shin:bhat:2018}. To that
end, we sort the models $\{\sg^{(i,k)}, 1 \le i \le m, 1 \le k \le N\}$ according to decreasing
posterior probabilities and retain the best (highest probability) $K$
models $\gamma^{(1)},\gamma^{(2)},\ldots, \gamma^{(K)}$ where $K$ is
chosen so that $f(\gamma^{(K)}|y)/f(\gamma^{(1)}|y) > \varepsilon$
where $\varepsilon$ is a prespecified tolerance (we use
$\log \varepsilon = -16$). Then we assign the weights
\begin{equation}\label{eq:weights}
 w_i = f(\gamma^{(i)}|y)/\sum_{k=1}^K f(\gamma^{(k)}|y)
\end{equation}
to the model $\gamma^{(i)}.$ We define the approximate marginal inclusion probabilities for the $j$th variable as $\hat\pi_j = \sum_{k=1}^K w_k\mathbb{I}(\gamma^{(k)}_j = 1)$ and define the weighted average model (WAM) as the model containing variables $j$ with $\hat\pi_j > 0.5.$  Note that if SVEN is allowed to run indefinitely to explore all $2^p$ models and $\varepsilon$ is set as zero, then the WAM would be theoretically identical to the median probability model \citep{barb:berg:2004}. However, computing the median probability model is infeasible when $p >> n$ because enumerating all the posterior probabilities of $\gamma$ is practically impossible.

In the literature, mostly the MAP (more precisely the \emph{discovered} MAP
model) model is used for prediction. In the
next section we develop methods for point and interval predictions using the top models $\gamma^{(k)}$'s with associated
weights $w_k$'s.


\section{Posterior predictive distribution and
  intervals} \label{sec:predicion} The posterior predictive
distribution of the response $y^*$ at a new covariate
vector $z^* \in \mathbb{R}^p,$ conditional on the observed covariate matrix $Z$ and
hyper-parameters $\lambda$ and $w$ is given by,
\begin{equation}\label{eq:postpreddist1}
  f(y^*|y) = \sum_\gamma \int_{\mathcal{S}_\gamma} f(y^*|z^*,\gamma, \mu_0,\mu_\gamma,\sigma^2) f(\gamma,\mu_0,\mu_\gamma,\sigma^2|y,Z) d\mu_0d\mu_\gamma d\sigma^2,
\end{equation}
where $f(y^*|z^*,\mu_0,\mu_{\gamma},\gamma,\sigma^2)$ is the density
of $\sN(\mu_0 +\mu_\gamma^\top z^*_\gamma,\sigma^2)$ as given in
\eqref{eq:regModelOriginalScale},
$f(\gamma,\mu_0,\mu_\gamma,\sigma^2 \linebreak|y,Z)$ is the joint posterior
density of $(\gamma,\mu_0,\mu_\gamma,\sigma^2)$ given $(y,Z)$
deduced from the hierarchical model \eqref{eq:litdiff}, and $\mathcal{S}_\gamma = (0,\infty)\times \mathbb{R}^{|\gamma|}\times \mathbb{R}$. Note that, the
distribution \eqref{eq:postpreddist1} is not 
tractable. However, as shown later in this section, posterior
predictive mean and variance of $y^*$ can be expressed as (posterior)
expectations of some analytically available functions of $\gamma$.
Also, samples from an approximation of \eqref{eq:postpreddist1} can be
drawn using our SVEN algorithm. Using these approaches, we now propose
two methods for computing approximate posterior prediction intervals
for $y^*.$

\subsection{A Z-prediction interval} \label{sec:posteriorPredSummary}
In this section we describe some approximations to $E(y^*|y)$ and $\Var(y^*|y)$ and use those to construct an interval for
$y^*$. To that end, from \eqref{eq:litdiff} we observe that $\beta_0$ and $\beta_\gamma$ are
conditionally independent given $y, \gamma, \sigma^2,$ and $Z$ with
\begin{equation}\label{eq:fullcondbeta}
 \beta_0|y,Z,\gamma,\sigma^2 \sim \mathcal{N}(\bar{y},\sigma^2/n),\;\textrm{and}\;\beta_\gamma|y,Z,\gamma,\sigma^2 \sim \mathcal{N}\left( A_{\gamma}^{-1} X^{\top}_{\gamma}\tilde{y}, \sigma^2 A_{\gamma}^{-1}\right),
\end{equation}
where $A_{\gamma}=X^{\top}_{\gamma} X_{\gamma} + \lambda I$ as defined
in section~\ref{sec:modeldes}.  Consequently, the full conditional
distribution of $(\mu_0,\mu_\gamma)$ is a $(|\gamma|+1)$-dimensional
multivariate Gaussian distribution given by
\begin{equation}\label{eq:fullcondmu}
 \begin{pmatrix}
 \mu_0\\ \mu_\gamma
  \end{pmatrix}
\bigg|\sigma^2,\gamma,y \sim \sN\left( 
 \begin{pmatrix}
  \bar{y} - \bar{Z}_\gamma^\top F_\gamma D_\gamma X_\gamma^\top\tilde{y}\\
  F_\gamma D_\gamma X_\gamma^\top\tilde{y}
 \end{pmatrix}
, \sigma^2
\begin{pmatrix}
 n^{-1} + \bar{Z}_\gamma^\top F_\gamma  \bar{Z}_\gamma & -  \bar{Z}_\gamma^\top F_\gamma \\ - F_\gamma \bar{Z}_\gamma & F_\gamma
\end{pmatrix}
\right),
\end{equation}
where $\bar{Z}_\gamma$ and $D_{\gamma}$ are sub-vector of $\bar{Z}$ and sub-matrix of $D$, respectively corresponding to the model $\gamma$, and $F_\gamma = D_\gamma^{-1}A_\gamma^{-1}D_\gamma^{-1}.$ Also,
\begin{equation}\label{eq:sigmaGivenGamma}
    \sigma^2|\gamma, y \sim \IG((n-1)/2, R_\gamma/2),
  \end{equation}
  where $\IG (a, b)$ denotes a inverse gamma random variable with
  density $f(\sigma^2) \propto (\sigma^2)^{-a-1} \exp(-b/\sigma^2)$, and
  $R_\gamma$ is defined in \eqref{eq:rgam}. Next, let  $\tilde{z}_\gamma = z_\gamma^* - \bar{Z}_\gamma$ and note that $\E(\sigma^2|\gamma, y)= R_\gamma/(n-3)$.  Thus, using iterated expectation and variance formulas, we have
  \begin{subequations}
    \label{eq:preddev}
  \begin{align}
  &\E(y^*|y) = \E\left[ \E\left\{y^*|\gamma, \sigma^2, \mu_0, \mu, y \right\}|y\right]  =   \E\left[\E\left\{\mu_0 + \mu_{\gamma}^\top z^*_{\gamma} |\sigma^2, \gamma, y\right \} |y\right] \nonumber \\
  & =  \bar{y} + \E\left[\left\{\tilde{z}_\gamma^\top F_\gamma D_\gamma X_\gamma^\top \tilde{y}\right\}|y\right]\quad \textrm{ and,} \label{eq:predmean}\\ 
  &   \Var(y^*|y) = \E\left(\Var\left\{y^*|\gamma, \sigma^2, \mu_0, \mu, y\right\}| y\right) + \Var\left(\E\left\{y^*|\gamma, \sigma^2, \mu_0, \mu, y \right\}|y\right)\nonumber \\  &= \E\left(\sigma^2|y\right) + \Var\left(\mu_0 + \mu_{\gamma}^\top z^*_{\gamma}|y\right)  
 \nonumber \\
 & =\E\left[\E(\sigma^2|\gamma,y)|y\right]+E\left[ \Var\left\{\mu_0+\mu_{\gamma}^\top z^*_{\gamma}|\sigma^2,\gamma,y\right\}|y\right] +\Var\left[\E(\mu_0+\mu_{\gamma}^\top z^*_{\gamma}|\sigma^2,\gamma,y)|y\right] \nonumber \\    
 & =   \E\left[ \dfrac{R_\gamma}{n-3}\left\{1 + \dfrac1n + \tilde{z}_\gamma^\top F_\gamma \tilde{z}_\gamma\right\}\bigg|y \right] + 
 \Var\left[ \left\{ \tilde{z}_\gamma^\top F_\gamma D_\gamma X_\gamma^\top \tilde{y} \right\}|y\right]\label{eq:predvar}
  \end{align}
  \end{subequations}
From \eqref{eq:predmean} and \eqref{eq:predvar} we see that
both $\E(y^*|y)$ and $\Var(y^*|y)$ can be expressed as posterior
expectations of analytically available functions of $\gamma$. However, because the posterior of $\gamma$ is not entirely available,
we propose using the models $\gamma^{(1)},\ldots,\gamma^{(K)}$
obtained from SVEN as described in section~\ref{sec:alg} with weights
$w_1,\ldots,w_K$ respectively, to approximate these expectations and
variances.  We can use these approximate posterior predictive mean and
variance of $y^*$ to obtain a $(1-\alpha)$ prediction interval for $y^*$ as
$\widehat{E}(y^*|y) \mp z_{\alpha/2}\widehat{\Var}(y^*|y)^{1/2},$ where $z_{\alpha/2}$ is
the $(1-\alpha/2)$th standard normal quantile. We call this
interval Z-prediction interval (Z-PI). Also, the posterior
predictive mean is used as a point estimate of $y^*$. In the next section, we describe an
alternative method for computing a prediction interval for $y^*$ using
Monte Carlo simulation.

\subsection{A Monte Carlo prediction interval}
A prediction interval for $y^*$ can also be constructed using Monte Carlo (MC) samples
generated from the posterior predictive distribution
\eqref{eq:postpreddist1}. Specifically, a $(1-\alpha)$ prediction
interval for $y^*$ is given by $\left[F_{y^*|y}^{-1}(\alpha/2),\ \  F_{y^*|y}^{-1}(1-\alpha/2)\right],$
where $F_{y^*|y}^{-1}(\alpha)$ denotes the $\alpha$-th quantile
of the distribution \eqref{eq:postpreddist1}. Now, we describe a
method for sampling from an approximation of \eqref{eq:postpreddist1}
using SVEN. To that end, we consider $\tilde{f}(y^*|y)$ given by
\begin{equation}\label{eq:apppostpreddist}
  \tilde{f}(y^*|y) = \sum_{i=1}^K w_i \int_{\mathcal{S}_{\gamma^{(i)}}} f(y^*|z^*,\gamma^{(i)}, \mu_0,\mu_{\gamma^{(i)}},\sigma^2) f(\mu_0,\mu_{\gamma^{(i)}},\sigma^2|\gamma^{(i)}, y,Z) d\mu_0d\mu_{\gamma^{(i)}} d\sigma^2,
\end{equation}\vspace{-0.1cm}
where $w_i$'s are defined in \eqref{eq:weights}, and
$\gamma^{(1)},\gamma^{(2)},\ldots, \gamma^{(K)}$ are the $K$ highest
probability models obtained by SVEN as described in
section~\ref{sec:alg}. Thus, $\tilde{f}(y^*|y)$ is the posterior
predictive pdf $f(y^*|y)$ given in \eqref{eq:postpreddist1} except
that the marginal posterior of $\gamma$ is replaced by a mixture
distribution of models chosen by SVEN. Samples from
\eqref{eq:apppostpreddist} can be drawn as follows. First, we sample
$\gamma$ from the top $K$ models with
$P(\gamma = \gamma^{(k)}) = w_k,$ $(1\leq k\leq K).$ Given $\gamma,$
we then sample $\sigma^2$ from \eqref{eq:sigmaGivenGamma}. Next given
$\gamma$ and $\sigma^2,$ we sample $\beta_0$ and $\beta_\gamma$ from
\eqref{eq:fullcondbeta}. Then we compute
$\mu_\gamma = D_\gamma^{-1} \beta_\gamma$ and
$\mu_0 = \beta_0 - \bar{Z}^{\top}_\gamma \mu_\gamma$, which are
samples from \eqref{eq:fullcondmu}. Finally generate $y^*$ from
$\mathcal{N}\left(\mu_0 + {\mu}_{\gamma}^\top{z}^*_{\gamma}, \sigma^2
\right).$ We repeat the above process a large number of times and
construct a $(1-\alpha)$ MC prediction interval (MC-PI) for $y^*$ as
$\left[\tilde{F}^{-1}(\alpha/2),\ \
  \tilde{F}^{-1}(1-\alpha/2)\right],$ where $\tilde{F}^{-1} (\cdot)$
denotes the empirical quantiles based on these samples. In practice,
generally one wants prediction intervals at several new covariate
vectors $z^*$'s. In section~\ref{sec.multpred} of the supplementary
materials, we describe a computationally efficient way of drawing
multiple samples from \eqref{eq:apppostpreddist} using the above
method and thus simultaneously computing prediction intervals at
several new covariate vectors $z^*$'s.


\section{Simulation studies}\label{sec:simulation}
In this section, we study the performance of our SVEN method through
several numerical experiments, and compare it with some other existing
methods. The competing variable selection methods we consider are S5
(\texttt{R} package: \texttt{BayesS5}), EMVS (\texttt{R} package:
\texttt{EMVS}), fastBVSR and three penalization methods, LASSO, Elastic
Net with elastic mixing parameter $\alpha = 0.5$ (\texttt{R} package:
\texttt{glmnet}) and SCAD (\texttt{R} package:
\texttt{ncvreg}). As also noted in \citet{shin:bhat:2018}, we could not
include BASAD \citep{nari:he:2014} for its high computational burden
and our ultra-high dimensional examples. As used in Table 1 of \citet{rock:geor:2014} we run EMVS with $v_1 = 1000$ and three choices for $v_0,$ namely, $v_0=0.2$ (EMVS$^{1}$), $v_0=0.6$ (EMVS$^{2}$) and $v_0=1$ (EMVS$^{3}$). For fastBVSR, the results are obtained using 100,000 MCMC iterations after a burn-in of 10,000 steps. For S5 the hyperparameters are tuned using a function provided in
\texttt{BayesS5}. Moreover, we denote by piMOM and peMOM,
respectively, the product inverse-moment and the product exponential
moment non-local priors used under S5. In addition, for piMOM and
peMOM, we use MAP and LS to denote the MAP estimator and the least
squares estimator from the MAP model, respectively. Under SVEN, both
MAP and WAM models, as described in section~\ref{sec:methodology} are
considered. For SVEN, we use $N=200$ and the temperature schedule described in Section \ref{sec:alg} with $m=9.$ Also, for SVEN, the ridge estimator
$\tilde{\beta}_{\gamma}$ is used to estimate the regression
coefficients for the MAP and the WAM models.

\subsection{Setup of experiments}
Our numerical studies are conducted in six different simulation settings described below.

\subsubsection{Independent predictors}\label{sec:sim.ind}
\noindent In this example, entries of  $X$ are generated independently from $\sN(0,1)$. The coefficients are specified as 
$\beta_1=0.5, \beta_2=0.75, \beta_3=1, \beta_4=1.25, \beta_5=1.5,$  and $\beta_j=0, \forall j > 5.$

\begin{table}
\begin{center}
\small\addtolength{\tabcolsep}{-2pt}
\caption{Independent predictors (Section \ref{sec:sim.ind})}
\label{iid}
\begin{tabular}{c|ccccccc}
	\hline
	\\[-1em]
	\textbf{Method} & \textbf{MSPE} & \textbf{MSE}\bm{$_\beta$} & \shortstack{\textbf{Coverage} \\ \textbf{probability (\%)}} & \shortstack{\textbf{Average} \\ \textbf{model size}} & \textbf{FDR (\%)} &\textbf{FNR (\%)} & \shortstack{\textbf{Jaccard} \\ \textbf{Index (\%)}}\\
	\\[-1em]
	\hline
	\\[-1em]
	SVEN(WAM) & 0.6387 & 0.0083 & 100 & 5 & 0 & 0 & 100\\
	\\[-1em]
	SVEN(MAP) & 0.6387 & 0.0083 & 100 & 5 & 0 & 0 & 100\\
	\\[-1em]
	piMOM(MAP) & 0.6384 & 0.0081 & 100 & 5 & 0 & 0 & 100\\
	\\[-1em]
	peMOM(MAP) & 0.6384 & 0.0080 & 100 & 5 & 0 & 0 & 100\\
	\\[-1em]
	piMOM(LS) & 0.6387 & 0.0083 & 100 & 5 & 0 & 0 & 100\\
	\\[-1em]
	peMOM(LS) & 0.6387 & 0.0083 & 100 & 5 & 0 & 0 & 100\\
	\\[-1em]
    fastBVSR & 0.6478 & 0.0091 & 100 & 5.09 & 1.45 & 0 & 98.55 \\
    \\[-1em]
	EMVS$^1$ &  1.0087 & 0.3777 & 0 & 3.80 & 0 & 24 & 76\\
	\\[-1em]
	EMVS$^2$ &  2.5203 & 1.8734 & 0 & 1.99 & 0 & 60.2 & 39.8\\
	\\[-1em]
	EMVS$^3$ &  5.0909 & 4.3994 & 0 & 0.53 & 0 & 89.4 & 10.6\\
	\\[-1em]
	LASSO & 0.7489 & 0.1146 & 100 & 56.5 & 87.34 & 0 & 12.66\\
	\\[-1em]
	SCAD & 0.6454 & 0.0152 & 100 & 18.42 & 47.50 & 0 & 52.50\\
	\\[-1em]
	Elastic Net & 0.8266 & 0.1898 & 100 & 91.15 & 93.08 & 0 & 6.92\\
  
	\hline
\end{tabular}
\end{center}
\end{table}

\begin{table}
\begin{center}
\small\addtolength{\tabcolsep}{-2pt}
\caption{Compound symmetry (Section \ref{sec:sim.cs}) with $\rho=0.6$.}\label{comp}
\begin{tabular}{c|ccccccc}
	\hline
	\\[-1em]
	\textbf{Method} & \textbf{MSPE} & \textbf{MSE}\bm{$_\beta$} & \shortstack{\textbf{Coverage} \\ \textbf{probability (\%)}} &\shortstack{\textbf{Average} \\ \textbf{model size}} & \textbf{FDR (\%)} &\textbf{FNR (\%)} & \shortstack{\textbf{Jaccard} \\ \textbf{Index (\%)}}\\
	\\[-1em]
	\hline
	\\[-1em]
	SVEN(WAM) & 48.3069 & 1.1912 & 100 & 5 & 0 & 0 & 100 \\
	\\[-1em]
	SVEN(MAP) & 48.3069 & 1.1892 & 100 & 5 & 0 & 0 & 100 \\
	\\[-1em]
	piMOM(MAP) & 48.2277 & 1.0018 & 100 & 5 & 0 & 0 & 100 \\
	\\[-1em]
	peMOM(MAP) & 50.1528 & 3.5669 & 94 & 4.96 & 0.37 & 1.2 & 98.5 \\
	\\[-1em]
	piMOM(LS) & 48.3069 & 1.1892 & 100 & 5 & 0 & 0 & 100 \\
	\\[-1em]
	peMOM(LS) & 50.2789 & 3.8758 & 94 & 4.96 & 0.37 & 1.2 & 98.5 \\
	\\[-1em]
    fastBVSR & 50.0479 & 2.5620 & 100 & 5.78 & 9.54 & 0 & 90.46 \\
    \\[-1em]
	EMVS$^1$ & 50.7090 & 7.0499 & 100 & 5.63 & 9.22 & 0 & 90.78 \\
	\\[-1em]
	EMVS$^2$ & 49.9839 & 5.3218 & 100 & 5.26 & 4.14 & 0 & 95.86 \\
	\\[-1em]
	EMVS$^3$ & 49.6243 & 4.5157 & 100 & 5.08 & 1.33 & 0 & 98.67 \\
	\\[-1em]
	LASSO & 55.2280 & 17.9975 & 100 & 51.02 & 89.94 & 0 & 8.44\\
	\\[-1em]
	SCAD & 48.3167 & 1.2556 & 100 & 6.29 & 11.55 & 0 & 88.45\\
	\\[-1em]
	Elastic Net & 57.5750 & 23.9724 & 100 & 89.68 & 93.76 & 0 & 6.24\\
	\hline
\end{tabular}
\end{center}
\end{table}

\begin{table}
\begin{center}
\small\addtolength{\tabcolsep}{-2pt}
\caption{Autoregressive correlation (Section \ref{sec:sim.ar}) with $\rho=0.6$.}
\label{ar}
\begin{tabular}{c|ccccccc}
	\hline
	\\[-1em]
	\textbf{Method} & \textbf{MSPE} & \textbf{MSE}\bm{$_\beta$} & \shortstack{\textbf{Coverage} \\ \textbf{probability (\%)}} & \shortstack{\textbf{Average} \\ \textbf{model size}} & \textbf{FDR (\%)} &\textbf{FNR (\%)} & \shortstack{\textbf{Jaccard} \\ \textbf{Index (\%)}}\\
	\\[-1em]
	\hline
	\\[-1em]
	SVEN(WAM) & 2.1521 & 0.0173 & 100 & 3 & 0 & 0 & 100 \\
	\\[-1em]
	SVEN(MAP) & 2.1521 & 0.0173 & 100 & 3 & 0 & 0 & 100 \\
	\\[-1em]
	piMOM(MAP) & 2.1519 & 0.0172 & 100 & 3 & 0 & 0 & 100 \\
	\\[-1em]
	peMOM(MAP) & 2.1515 & 0.0168 & 100 & 3 & 0 & 0 & 100 \\
	\\[-1em]
	piMOM(LS) & 2.1521 & 0.0173 & 100 & 3 & 0 & 0 & 100 \\
	\\[-1em]
	peMOM(LS) & 2.1521 & 0.0173 & 100 & 3 & 0 & 0 & 100 \\
	\\[-1em]
    fastBVSR & 2.1961 & 0.0187 & 100 & 3.03 & 0.75 & 0 & 99.25 \\
    \\[-1em]
	EMVS$^1$ & 2.2738 & 0.1286 & 100 & 6.7 & 54.57 & 0 & 45.43 \\
	\\[-1em]
	EMVS$^2$ & 2.2803 & 0.1419 & 100 & 5.28 & 41.42 & 0 & 58.58 \\
	\\[-1em]
	EMVS$^3$ & 2.2947 & 0.1619 & 100 & 4.33 & 28.40 & 0 & 71.60 \\
	\\[-1em]
	LASSO & 2.3118 & 0.1641 & 100 & 28.16 & 76.82 & 0 & 23.19\\
	\\[-1em]
	SCAD & 2.1592 & 0.0252 & 100 & 10.33 & 28.30 & 0 & 71.70\\
	\\[-1em]
	Elastic Net & 2.4590 & 0.3754 & 100 & 54.35 & 91 & 0 & 9.00\\
	\hline
\end{tabular}
\end{center}
\end{table}

\begin{table}
\small\addtolength{\tabcolsep}{-2pt}
\caption{Group structure with 3 groups (Section \ref{sec:sim.group}). }
\label{group}
{\centering
\begin{tabular}{c|ccccccc}
	\hline
	\\[-1em]
	\textbf{Method} & \textbf{MSPE} & \textbf{MSE}\bm{$_\beta$} & \shortstack{\textbf{Coverage} \\ \textbf{probability (\%)}} & \shortstack{\textbf{Average} \\ \textbf{model size}} & \textbf{FDR (\%)} &\textbf{FNR (\%)} & \shortstack{\textbf{Jaccard} \\ \textbf{Index (\%)}}\\
	\\[-1em]
	\hline
	\\[-1em]
	SVEN(WAM)$^4$ & 78.7067 & 299.4512 & 0 & 2.65 & 0 & 82.33 & 17.67\\
	\\[-1em]
	SVEN(MAP)$^4$ & 81.0355 & 533.5387 & 0 & 3 & 0 & 80 & 20\\
	\\[-1em]
	SVEN(WAM)$^5$ & 82.5443 & 1.8816 & 98 & 14.99 & 0.06 & 0.13 & 99.80\\
	\\[-1em]
	SVEN(MAP)$^5$ & 82.1825 & 1.6467 & 98 & 15.02 & 0.25 & 0.13 & 99.62\\
	\\[-1em]
	piMOM(MAP) & 81.3345 & 528.8252 & 0 & 3.02 & 0.4 & 80 & 19.98\\
	\\[-1em]
	peMOM(MAP) & 81.7316 & 530.0427 & 0 & 3.02 & 0.4 & 80 & 19.98\\
	\\[-1em]
	piMOM(LS) & 81.2392 & 528.7916 & 0 & 3.02 & 0.4 & 80 & 19.98\\
	\\[-1em]
	peMOM(LS) & 81.6289 & 530.1160 & 0 & 3.02 & 0.4 & 80 & 19.98\\
	\\[-1em]
    fastBVSR & 81.1029 & 326.776 & 0 & 4.14 & 1.38 & 72.87 & 27.01 \\
    \\[-1em]
	EMVS$^1$ & 79.0816 & 54.5117 & 86 & 15.20 & 2.07 & 0.93 & 97.03 \\
	\\[-1em]
	EMVS$^2$ & 77.8038 & 14.9534 & 99 & 15.05 & 0.38 & 0.07 & 99.56 \\
	\\[-1em]
	EMVS$^3$ & 77.5867 & 7.5430 & 100 & 15.02 & 0.13 & 0 & 99.88 \\
	\\[-1em]
	LASSO & 84.9837 & 111.852 & 0 & 9.36 & 63.49 & 28.93 & 29.96\\
	\\[-1em]
	SCAD & 81.2506 & 530.2818 & 0 & 11.59 & 30.54 & 80 & 16.28\\
	\\[-1em]
	Elastic Net & 85.7453 & 9.3598 & 100 & 68.03 & 65.94 & 0 & 34.06\\
	\hline
\end{tabular}}\\
{$^4\lambda=n/p^2$, $w=\sqrt{n}/p$; $^5\lambda=200$, $w=0.02$.}
\end{table}

\begin{table}
\begin{center}
\small\addtolength{\tabcolsep}{-2pt}
\caption{Factor model with 2 factors (Section \ref{sec:sim.factor}).}
\label{factor}
\begin{tabular}{c|ccccccc}
	\hline
	\\[-1em]
	\textbf{Method} & \textbf{MSPE} & \textbf{MSE}\bm{$_\beta$} & \shortstack{\textbf{Coverage} \\ \textbf{probability (\%)}} & \shortstack{\textbf{Average} \\ \textbf{model size}} & \textbf{FDR (\%)} &\textbf{FNR (\%)} & \shortstack{\textbf{Jaccard} \\ \textbf{Index (\%)}}\\
	\\[-1em]
	\hline
	\\[-1em]
	SVEN(WAM) & 42.9106 & 0.3892 & 100 & 5 & 0 & 0 & 100 \\
	\\[-1em]
	SVEN(MAP) & 42.9103 & 0.3891 & 100 & 5 & 0 & 0 & 100 \\
	\\[-1em]
	piMOM(MAP) & 42.8731 & 0.3724 & 100 & 5 & 0 & 0 & 100 \\
	\\[-1em]
	peMOM(MAP) & 42.9491 & 0.4211 & 100 & 5.01 & 0.17 & 0 & 99.83 \\
	\\[-1em]
	piMOM(LS) & 42.9103 & 0.3891 & 100 & 5 & 0 & 0 & 100 \\
	\\[-1em]
	peMOM(LS) & 42.9361 & 0.4083 & 100 & 5.01 & 0.17 & 0 & 99.83 \\
	\\[-1em]
    fastBVSR & 67.0982 & 19.9837 & 87 & 6.14 & 18.52 & 3.60 & 79.89 \\
    \\[-1em]
	EMVS$^1$ & 64.6038 & 22.1115 & 95 & 19.13 & 66.40 & 1.00 & 33.59 \\
	\\[-1em]
	EMVS$^2$ & 56.7884 & 14.5042 & 95 & 11.58 & 45.34 & 1.00 & 54.64 \\
	\\[-1em]
	EMVS$^3$ & 53.4840 & 11.3980 & 94 & 9.08 & 34.73 & 1.20 & 65.20 \\
	\\[-1em]
	LASSO & 54.2887 & 11.2984 & 99 & 66.37 & 91.81 & 0.20 & 7.03\\
	\\[-1em]
	SCAD & 43.1155 & 0.5743 & 100 & 11.56 & 27.99 & 0 & 72.01\\
	\\[-1em]
	Elastic Net & 62.4327 & 19.4566 & 99 & 54.29 & 95.90 & 0.20 & 4.10\\
	\hline
\end{tabular}
\end{center}
\end{table}

\begin{table}
\begin{center}
\small\addtolength{\tabcolsep}{-2pt}
\caption{Extreme correlation (Section \ref{sec:sim.extreme}).}
\label{extrm}
\begin{tabular}{c|ccccccc}
	\hline
	\\[-1em]
	\textbf{Method} & \textbf{MSPE} & \textbf{MSE}\bm{$_\beta$} & \shortstack{\textbf{Coverage} \\ \textbf{probability (\%)}} & \shortstack{\textbf{Average} \\ \textbf{model size}} & \textbf{FDR (\%)} &\textbf{FNR (\%)} & \shortstack{\textbf{Jaccard} \\ \textbf{Index (\%)}}\\
	\\[-1em]
	\hline
	\\[-1em]
	SVEN(WAM) & 14.0754 & 0.1571 & 100 & 5 & 0 & 0 & 100 \\
	\\[-1em]
	SVEN(MAP) & 14.0757 & 0.1569 & 100 & 5 & 0 & 0 & 100 \\
	\\[-1em]
	piMOM(MAP) & 14.0732 & 0.1547 & 100 & 5 & 0 & 0 & 100 \\
	\\[-1em]
	peMOM(MAP) & 14.0750 & 0.1562 & 100 & 5 & 0 & 0 & 100 \\
	\\[-1em]
	piMOM(LS) & 14.0757 & 0.1569 & 100 & 5 & 0 & 0 & 100 \\
	\\[-1em]
	peMOM(LS) & 14.0757 & 0.1569 & 100 & 5 & 0 & 0 & 100 \\
	\\[-1em]
    fastBVSR & 31.2771   & 32.1993 & 97 & 6.55 & 18.65 & 0.6 & 81.03 \\
    \\[-1em]
	EMVS$^1$ & 14.7568 & 2.6871 & 100 & 5.6 & 8.44 & 0 & 91.56 \\
	\\[-1em]
	EMVS$^2$ & 14.4561 & 1.5340 & 100 & 5.09 & 1.45 & 0 & 98.55 \\
	\\[-1em]
	EMVS$^3$ & 14.4218 & 1.3793 & 100 & 5.03 & 0.5 & 0 & 99.5 \\
	\\[-1em]
	LASSO & 15.3893 & 2.8732 & 100 & 13.77 & 61.13 & 0 & 23.68\\
	\\[-1em]
	SCAD & 14.0799 & 0.1678 & 100 & 5.49 & 5.29 & 0 & 94.71\\
	\\[-1em]
	Elastic Net & 15.5365 & 3.7949 & 100 & 65.87 & 86.75 & 0 & 13.25\\
	\hline
\end{tabular}
\end{center}
\end{table}

\subsubsection{Compound symmetry}\label{sec:sim.cs}
This example is taken from Example 3 in \citet{wang:2009} and Example
2 in \citet{wang:leng:2016}. The rows of $X$ are generated independently from $\sN_p\left(0,(1-\rho)I_p + \rho1_p1_p^\top\right)$ where we take
$\rho = 0.6$. The regression coefficients are set as $\beta_j = 5$ for
$j = 1, \ldots, 5$ and $\beta_j=0$ otherwise.

\subsubsection{Auto-regressive correlation} \label{sec:sim.ar}
The auto-regressive correlation structure is commonly observed in time
series data where the correlation between observations depends on the
time lag between them. In this example, we use AR(1) structure where
the variables further apart from each other are less
correlated. Following Example 2 in \citet{wang:leng:2016},  
$X_j = \rho X_{j-1} + (1-\rho^2)^{1/2}z_j,$ for $1 \leq j \leq p,$
where $X_0$ and $z_j$ ($1\leq j \leq p)$ are iid $\sim \sN_n(0,I_n).$
We use $\rho=0.6$ and set the regression coefficients  as $\beta_1 = 3$, $\beta_4=1.5$, $\beta_7=2$ and $\beta_j=0$ for $j \not \in \{1,4,7\}$.

\subsubsection{Factor models} \label{sec:sim.factor}
This example is from \citet{mein:buhl:2006} and
\citet{wang:leng:2016}. With a fixed number of factors, $K$, we first generate a $p\times K$ matrix $F$ whose entries are iid standard normal. Then the rows of $X$ are independently generated from $\sN_p(0,FF^\top + I_p).$ We fix $K=2$ and the regression coefficients are set to be the same as in Example \ref{sec:sim.cs}.

\subsubsection{Group structure} \label{sec:sim.group}
This special correlation structure arises when variables are grouped
together in the sense that the variables from the same group are
highly correlated. This example is similar to \citet{wang:leng:2016}
and is similar to example 4 of Zou and Hastie (2005) where 15 true
variables are assigned to 3 groups. We generate the predictors as
$X_{m} = z_1 + \zeta_{1, m}$, $X_{5+m} = z_2 + \zeta_{2, m}$,
$X_{10+m} = z_3 + \zeta_{3, m}$  where
$z_i$ are iid $\sim \sN_n(0, I_n)$ and
$\zeta_{i, m} \overset{iid}\sim \sN_n(0, 0.01I_n)$ for $1\leq i \leq 3$ and for $m=0,1,2,3,4$. The regression
coefficients are set as $\beta_j=3$ for $j \in \{1, 2, \ldots, 15\}$
and $\beta_j=0$ otherwise.

\subsubsection{Extreme correlation} \label{sec:sim.extreme}
This challenging example is the example 6 of
\citet{wang:leng:2016}. In this example, We first simulate $z_j$ ,
$j=1, \ldots, p$ and $w_j$, $j=1, \ldots, 5$ independently from the
multivariate standard normal distribution $\sN_n(0, I_n)$. Then the covariates are
generated as $X_j=(z_j+w_j)/\sqrt{2}$ for $j=1, \ldots, 5$ and
$X_j=(z_j+\sum_{i=1}^5 w_i)/2$ for $j = 6, \ldots, p$. By setting the
number of true covariates to be 5 and let $\beta_j = 5$ for
$j = 1, \ldots, 5$ and $\beta_j=0$ for $j = 6, ..., p$, the
correlation between the response and the unimportant covariates is
around $2.5/\sqrt{3}$ times larger than that between the response and
the true covariates, making it difficult to identify the important 
covariates.

Our simulation experiments are conducted using 100 simulated pairs of
training and testing data sets. For each of the simulation settings
introduced above, we set $p=20,000$ and generate training data set and
testing data set of size $n=400$ each. The error variance $\sigma^2$ is
determined by setting theoretical $R^2 = 90\%$ \citep{wang:2009}. The hyperparameters $w$ and $\lambda$ are chosen to be $\sqrt{n}/p$ and $n/p^2$, respectively, except for group structure where we also use $\lambda = 200$ and $w=0.02$ to account for the high within-group correlation and relatively large true model size.

In order to evaluate the performance of the propose method, we compute
the following metrics: (1) mean squared prediction error based on
testing data (MSPE); (2) mean squared error between the estimated
regression coefficients and the true coefficients
(MSE$_{\bm{\beta}}$); (3) coverage probability which is defined as the
proportion of the selected models containing the true model (4)
average model size which is calculated as the average number of
predictors included in the selected models over all the replications
(5) false discovery rate (FDR); (6) false negative rate (FNR) and (7)
the Jaccard index which is defined as the size of the intersection divided
by the size of the union of the selected model and the true model. All computations are done using single--threaded R  on 
a workstation with two 2.6 GHz 8-Core Intel\textregistered E5-2640 v3 processors and 128GB RAM.

\subsection{Simulation results and main findings}
\label{sec:simu}


The average of the metrics of our simulation results are presented in Tables
\ref{iid}-\ref{extrm}. For peMOM and piMOM priors, the difference between the MAP and the LS only arise in the MSPE and the MSE$_{\bm{\beta}}$ but not in the other metrics. We can observe from the tables that SVEN
and S5 perform much better than EMVS, fastBVSR
and the three frequentist penalized methods in general. In most settings, the penalized
methods result in many false discoveries, yet attaining similar or
worse coverage probabilities compared to the Bayesian methods. Since
the estimates of $\beta$ from EMVS are not sparse, it has higher
MSE$_{\bm{\beta}}$ than SVEN and S5. As observed in Tables~\ref{factor} and \ref{extrm}, fastBVSR results in large values of MSPE and MSE$_{\bm{\beta}}$ due to poor estimates of $\beta$. Also, SVEN yields
competitive prediction errors and has better FDR and Jaccard indices
in every case other than the group structure.

For the case of group structure (Table \ref{group}) where there is a
high correlation between the variables within the same group, SVEN
with $w=\sqrt{n}/p$ and $\lambda=n/p^2$ and S5 both pick up only one
representative variable from each group, resulting in a high false
negative rate and average model size around three. Although elastic
net regression successfully includes all the important variables it
also includes a large number of unimportant variables and thus leads
to a very high false discovery rate. However, by increasing the
shrinkage to $\lambda=200$ and increasing the prior inclusion
probability to $w=0.02$, SVEN stands out from its competitors. In
fact, if important predictors are anticipated to be highly
correlated, this prior information can be incorporated by choosing a
larger value for $\lambda.$

In addition, we compare the computing times between S5 (with both piMOM and peMOM priors) and SVEN and find that SVEN hits the MAP model faster than S5. The details are provided in Section \ref{sec.comptime}. 



\section{Real data analysis}\label{sec:dataanalysis}
We examine the practical performance of our proposed method by
applying it to a real data example. \citet{cook:mcmu:2012} conducted a
genome-wide association study on starch, protein, and kernel oil
content in maize. The original field trial at Clayton, NC in 2006
consisted of more than 5,000 inbred lines and check varieties primarily
coming from a diverse IL panel consisting of 282 founding lines
\citep{flin:thui:2005}. Because the dataset comes from a field trial,
the responses could be spatially autocorrelated. Thus we use a random
row-column adjustment to obtain the adjusted phenotypes of the
varieties. However, marker information of only $n=3,951$ of these
varieties are available from the panzea project
(https://www.panzea.org/) which provide information on 546,034 single
nucleotide polymorphisms (SNP) markers after removing duplicates and
SNPs with minor allele frequency (MAF) less than 5\%. We use the
protein content as our phenotype for conducting the association
study. Because the inbred varieties are bi-allelic, we store the
marker information in a sparse format by coding the minor alleles by
one and major alleles by zero.

\begin{figure}[htp]
    \centering
.    \includegraphics[width=0.7\textwidth]{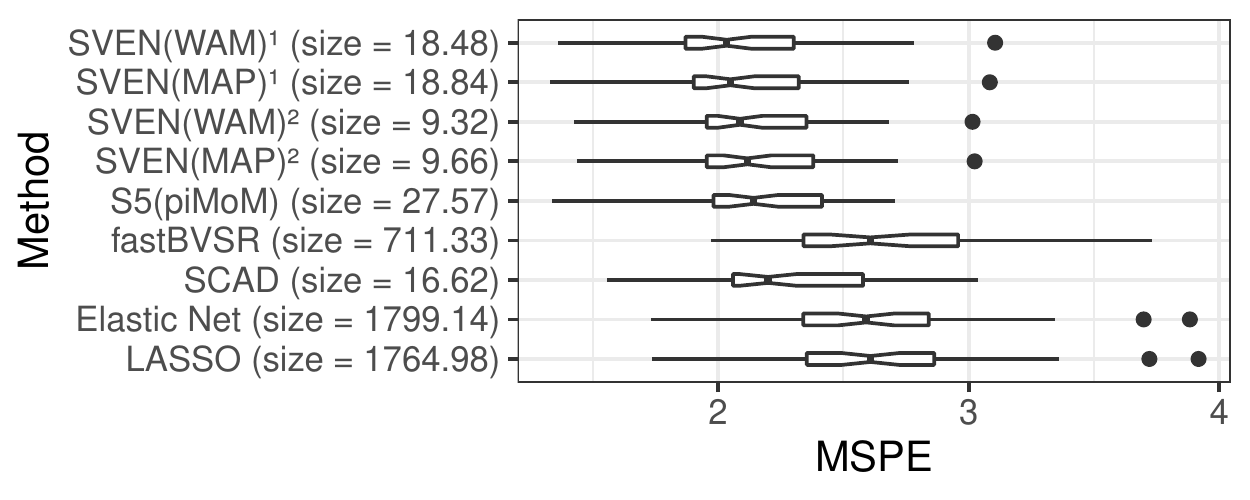}
    \caption{Boxplots for MSPE using SVEN, S5, fastBVSR, LASSO and Elastic net after screening. $^1 w=1/p_0, \lambda=\sqrt{n_0}$; $^2 w=1/p_0, \lambda=n_0/p_0^2.$}
    \label{mspe.scr}
\end{figure}

\subsection{Marker selection after screening}\label{sec:data-screening}
  We compare our method to S5, fastBVSR and the three penalized regression methods
(LASSO, Elastic Net and SCAD). Since both R packages \texttt{BayesS5}
(version 1.31) and \texttt{glmnet} (version 2.0-18) do not work on
this massive data set, we perform a screening of these markers before
conducting variable selection so as to reduce the dimension of the
data. We randomly split the data into a training set of size $n_0 = 3,851$ and
testing set of size 100. Then we use high dimensional ordinary least
squares projection (HOLP) screening method \citep{wang:leng:2016} to
preserve $p_0=3,851$ markers. Note that the training sets are formed by
controlling the MAF of each marker to be no less than 1.5\%. Because markers tend
to be highly correlated, we use SVEN with $w=1/p_0$ but with two choices of $\lambda:$ $\lambda = \sqrt{n_0}$
(high shrinkage) and $\lambda=n_0/p_0^2$ (low shrinkage); and with $m=3$ and $N=50$ for selecting the markers. In our experience, both the model size and MSPE lie in between the respective reported values for other intermediate values of $\lambda$ that we have tried.
We repeat the entire process 50 times -- each time
computing the MSPE and the model size from each method. The peMOM
non-local prior in S5 failed to provide any result even after 100
hours of running, and S5 with the piMOM prior failed to provide a
result in three cases. The fastBVSR algorithm ran successfully in only 39 out of the 50 cases, while the complex iterative factorization at the core of fastBVSR encountered floating point errors in the remaining 11 cases and could not produce any result. In contrast, SVEN faced no difficulties and produced the results within reasonable time.

The boxplots of these MSPEs are shown in Figure \ref{mspe.scr} along
with the average model sizes. Overall SVEN, S5 and SCAD perform significantly better than the lasso, the elastic net regression and fastBVSR and produce smaller MSPE with smaller model sizes. Moreover, SVEN and S5
produce comparable MSPE values but SVEN results in more parsimonious
models. SVEN with high shrinkage produces slightly smaller MSPE but double
model size than with low shrinkage.

\subsection{Marker selection on the entire data set}
Unlike other variable selection methods, SCAD and SVEN can be successfully
directly applied to the whole data set without any pre-screening. We ran
the SVEN 50 times again with the temperature schedule described in Section \ref{sec:alg} with $m=3,$
with $N=100$ iterations per temperature, each time starting with a different random seed.  Initially, we use 
$w=1/p$ and try several values of $\lambda$ as done in Section \ref{sec:data-screening}.
The best models from these
50 runs vary suggesting the posterior surface is severely
multimodal. With $\lambda=n/p^2,$ we find that although the sizes of these best models
remain around nine, the number of unique markers included in at least
one of these 50 best models is over 30 (for SCAD these numbers were $>40$ and $>60,$ respectively). Other larger values of $\lambda$ produce even larger models and more unique variables.
Interestingly, by taking a further look into the markers it identified, we
discovered that the presence of some of these markers in a model is always
accompanied by the absence of certain other markers. More
specifically, some pairs and triplets of the markers are never
included simultaneously in the MAP models but the frequencies at
which they are selected add up to 50.
Thus to achieve more parsimonious models, we reduce $w$ to $1/p^2$ and use $\lambda=n/p^2.$ 
Using such a small $w$, the sizes of the
best models from each run reduce to around four and the number of
unique markers that are included at least once in the 50 best models
comes down to eight. To verify our conjecture on the correlations
between these markers, we calculated the pairwise partial correlations
between these eight markers. It turns out that the pairs of markers that
are never included in the same model are indeed relatively highly
partially correlated than other pairs. Figure \ref{graph} gives the
partial correlations for those markers where the size of the nodes
indicates the number of times the markers are included in one of the
50 best models and Pairs of markers that are never included or excluded jointly
are joined by a line segment. Note that the partial
correlation between the connected markers are at least 29\% whereas the largest partial correlation for markers that are not connected is around 18\%. The inclusion frequencies
of the pairs of connected markers add up to 50. Note that the fifth and sixth important
markers are not grouped with other markers because their inclusions or
exclusions are not related with the inclusion or exclusion of any
other marker. Thus SVEN is able to identify pairs of markers that have
similar effect on the response.

\begin{figure}[t]
    \centering
    \includegraphics{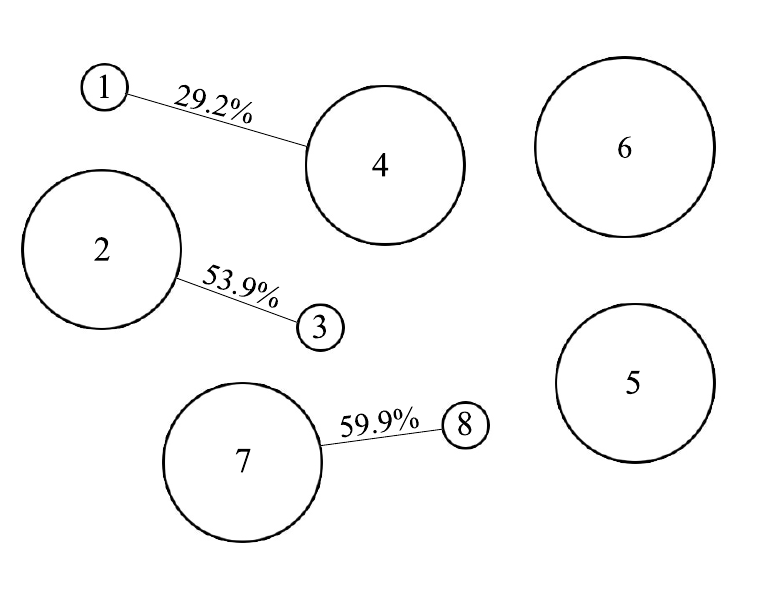}
    \caption{Graph for the selected markers and their corresponding
      partial correlations using $w=1/p^2$ and $\lambda=n/p^2$. The SNP accession
      numbers of the selected markers are:
      1=5\_151885291, 2= 5\_197591528, 3=5\_200552088, 4=6\_7585863,
      5=7\_153216557, 6=9\_142949160, 7=10\_72608193, and 
      8=10\_110298386.}
    \label{graph}
\end{figure}

\begin{figure}[h]
    \centering
    \includegraphics[width=0.9\textwidth]{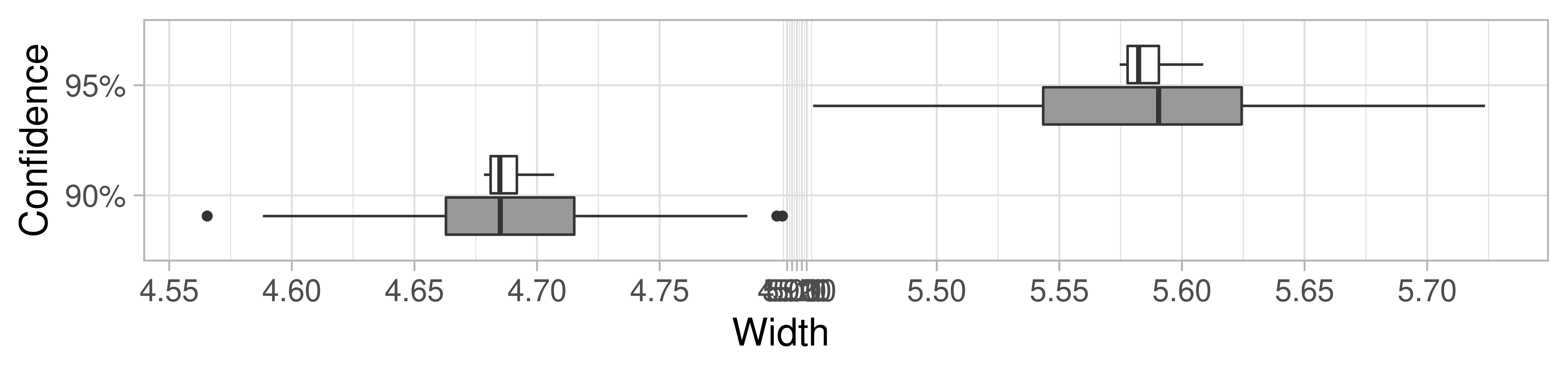}
    \caption{Boxplots of the widths of MC-PIs (grey) and Z-PIs (white).}
    \label{bplot.width}
\end{figure}
Next, we study the performance and the widths of the 90\% and 95\% Z-PIs and MC-PIs described in Section \ref{sec:predicion}.
To that end, we randomly split the entire data into a
training set of size $n=3,751$ under the constraint that the MAF of each marker is at least 1.5\% and a testing set of size $200.$ We also remove any duplicated markers from the training set, which results in a
smaller $p = 544,211$. We generate 10,000 samples from the approximate
posterior predictive distribution \eqref{eq:apppostpreddist} to compute the MC-PIs. We find that the Z-PIs and the
MC-PIs attain identical coverage rates and these are found to be 91\% and 95\% for the 90\% and 95\% prediction
intervals, respectively. The boxplots of the widths of the 200 intervals from each method
are presented in Figure \ref{bplot.width}. We find that widths of the
the Z-PIs are less variable compared to the same for the
MC-PIs. It is encouraging to see that despite
non-normality of the posterior prediction distribution, the Z-PIs are better than simulation based intervals.

\section{Conclusion}\label{sec:conclusion}
In this article, we introduce a Bayesian variable selection method
with embedded screening for ultrahigh-dimensional settings. The model
used here is a hierarchical model with well known spike and slab
priors on the regression coefficients. Use of the degenerate spike
prior for inactive variables not only results in sparse estimates of
regression coefficients and (much) lesser computational burden, it
also allows us to establish strong model selection consistency under
somewhat weaker conditions than \cite{nari:he:2014}. In particular, we
prove that the posterior probability of the true model converges to
one even when the norm of mean effects solely due to the unimportant
variables diverge. On the other hand, our method crucially hinges on
the fact that model probabilities are available in closed form (up to a
normalizing constant) which is due to the use of Gaussian slab priors
on active covariates. We propose a scalable variable selection
algorithm with an inbuilt screening method that efficiently explores
the huge model space and rapidly finds the MAP model. The screening is
actually model based in the sense that it is performed on a set of
candidate models rather than the set of potential variables.  The
algorithm also incorporates the temperature control into a neighbor
based stochastic search method. We use fast Cholesky update to
efficiently compute the (unnormalized) posterior probabilities of the neighboring
models. Since mean and variance of the posterior predictive
distribution are shown to be means of analytically available functions
of the models, a derivative of the proposed method is construction of
novel prediction intervals for future observations. Both Z based
intervals and simulation based intervals are derived and compared.  In
the context of the real data analysis, we observe that Z based
prediction intervals lead to the same coverage rates, although are
narrower than Monte Carlo intervals.  The extensive simulations
studies in section~\ref{sec:simulation} and the real data analysis in
section~\ref{sec:dataanalysis} demonstrate the superiority of the
proposed method compared with the other state of the art methods, even
though the hyperparameters in the proposed method are not carefully
tuned. Among the Bayesian methods used for comparison, the package associated
with the proposed algorithm seems to be the only one that can be
directly applied to datasets of dimension as high as the one analyzed here
with the computing resources mentioned before.

Based on the Cholesky update described in Section \ref{sec:alg}, SVEN can be extended to accommodate the determinantal point process prior \citep{koji:koma:2016} on $\gamma$ given by $p(\gamma|\omega) \propto \omega^{|\gamma|}\left|X^\top_\gamma X_\gamma\right|,$ where $\omega > 0.$
Variable selection and consistency of the resulting posteriors for
high dimensional generalized linear models are considered in
\cite{lian:song:yu:2013}. It would be interesting to extend our method
to the generalized linear regression model setup. The dataset we have
used comes from an agricultural field trial and hence the observations
are expected to be spatially autocorrelated. Although we have used a
two stage procedure by first obtaining spatially adjusted genotypic
effects, our model can be extended to include spatial random effects
\citep{dutt:mond:2014}. Also, in many applications, the covariates may
have a non-linear effect on the response and our method could be
extended to additive models.

\paragraph{Supplemental materials} The supplemental materials contain additional details on computations and proofs of the theoretical results stated in the paper.

\begingroup
\small
\setstretch{1.0}
\bibliographystyle{asadoi}
\bibliography{reference}
\endgroup

\pagebreak
\begin{center}
\textbf{\large Supplement to \\
``Model Based Screening Embedded Bayesian Variable Selection for Ultra-high Dimensional Settings" \\Dongjin Li, Somak Dutta and Vivekananda Roy}
\end{center}
\setcounter{equation}{0}
\setcounter{figure}{0}
\setcounter{table}{0}
\setcounter{page}{1}
\setcounter{section}{0}
\makeatletter
\renewcommand{\thesection}{S\arabic{section}}
\renewcommand{\thesubsection}{\thesection.\arabic{subsection}}
\renewcommand{\theequation}{S\arabic{equation}}
\renewcommand{\thefigure}{S\arabic{figure}}
\renewcommand{\bibnumfmt}[1]{[S#1]}
\renewcommand{\citenumfont}[1]{S#1}
\renewcommand{\thetable}{S\arabic{table}}

\section{Efficient computations for multiple predictions}
\label{sec.multpred}

We describe in this section how we efficiently generate multiple $y^*$
in order to obtain the empirical posterior predictive distribution and
compute the Monte Carlo prediction intervals at several new covariates $z^{*(1)}, \dots, z^{*(L)}$. Recall from Section
\ref{sec:alg} that for a model $\gamma$, $U_\gamma$ is the upper
triangular Cholesky factor of $X_\gamma^\top X_\gamma + \lambda I$ and
$v_\gamma = U_\gamma^{-\top}X_\gamma^\top \tilde{y}.$ The detailed
procedure is described below.
\begin{algorithm}[H]
\setstretch{1}
\caption{Generate multiple $y^*$}
\label{alg:pred}
\begin{algorithmic}[1]
    \State Sample $N$ models with replacement from the best $K$ models returned by SVEN,  with probabilities proportional to $w_i$ defined in \eqref{eq:weights} for $i = 1, \ldots, K$
    \State From the models sampled from step 1, find the unique models $\gamma^1, \ldots, \gamma^M$ such that $\sum_{m=1}^{m=M}S_m = N$, where $S_m$ denote the number of models identical to $\gamma^m$ 
    \State Compute $U_{\gamma^m}$ and $v_{\gamma^m}$ for $m=1, \ldots, M$
    \For{$m=1$ to $m=M$}
        \For{$j=1$ to $j=S_m$}
        \State Sample $\sigma^2$ from $\IG\left((n-1)/2, \text{R}_{{\gamma^m}}/2\right)$
        \State Sample $e_i$ from $\mathcal{N}(0, \sigma^2)$ for $i=1,\ldots,|\gamma^m|$
        \State Compute $\mu_{\gamma^m} = D_{\gamma^m}^{-1}\left(U_{\gamma^m}^{-1}(v_{\gamma^m}+e)\right)$, where $e = (e_1, \ldots, e_{|\gamma^m|})^{\top}$
        \State Sample $\mu_0$ from $\mathcal{N}(\bar{y}-\bar{Z}_{\gamma^m}^\top \mu_{\gamma^m},\sigma^2/n)$
                \For{$\ell=1$ to $\ell=L$}
                \State Generate $y^{*}$ from $\mathcal{N}\left(\mu_0 + z^{*(\ell)}_{\gamma^m}{\mu}_{\gamma^m}, \sigma^2 \right)$
                \EndFor
        \EndFor
    \EndFor
\end{algorithmic}
\end{algorithm}

\section{Comparison of computation time}
\label{sec.comptime}
We examine the computation time it takes for SVEN to hit the MAP model for the first time, and compare it with S5 under the piMOM and the peMOM priors. We simulate the data according to Section \ref{sec:sim.ar}, where $Z$ has an AR(1) structure. We consider five different $(n, p)$ pairs with $p=2n^{3/2}$ where $n \in \{100, 225, 400, 625, 900\}.$ For each of the $(n, p)$ pair, we obtain the computation times over 10 replicates. For SVEN, we use $w=\sqrt{n/p}$, $\lambda=n/p^2$ and $N=50$ with the temperature schedule described in Section \ref{sec:alg} with $m=3.$ Again, S5 is implemented using R-package \texttt{BayesS5} using their default tuning parameter with only one repetition.

Figure \ref{timeplot} shows the median computation times SVEN and S5 take to first hit the MAP model, excluding the preprocessing steps which are negligible. Both algorithms attain the same MAP model for all the data sets. In general, SVEN hits the MAP model faster than  S5 for both small and large number of variables. 
Moreover, compared to S5, the computation time for SVEN increases at a slower rate as $p$ gets larger.

\begin{figure}[H]
    \centering
    \includegraphics[width=0.7\textwidth]{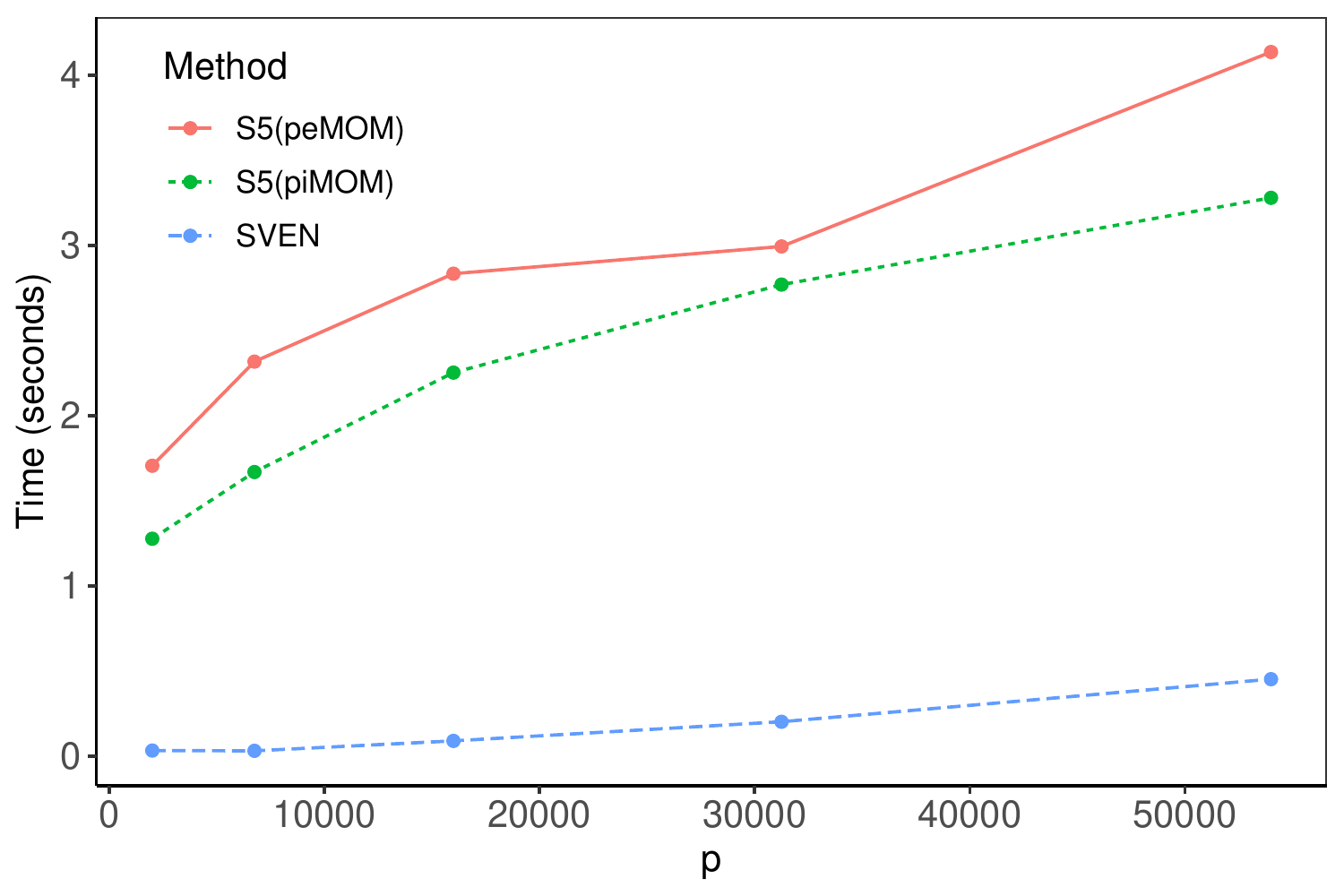}
    \caption{The median computation time to first hit the MAP model for SVEN and S5.}
    \label{timeplot}
\end{figure}

\section{Preliminary results}
Let $R^{*}_{\gamma} = \tilde{Y}^{\top}(I-P_{\gamma})\tilde{Y}$ which is the residual
sum of squares obtained by ordinary least squares and also let
$Q_{\gamma} = \lambda^{|{\gamma}|/2}|X^{\top}_{\gamma}
X_{\gamma}+\lambda I|^{-1/2}$. Before proving the model selection
consistency stated in Theorems \ref{thm.consistency} and
\ref{thm.consistency.nosigma}, we first provide some preliminary
results on the bound of $Q_{\gamma}/Q_t$ which will be used to
bound the posterior ratio of a given model $\gamma$ to the true model
$t$, and the bound of the difference between $R_t$ and $R^*_t$.
\begin{lemma}\label{lm.Qgamma}
  For any model $\gamma \not = t$,
  $\frac{Q_{\gamma}}{Q_t} \leq v^{'} (n
  \eta_m^n(\nu)/\lambda)^{-(r^{*}_{\gamma} - r_t)/2}
  (\eta_m^n(\nu))^{-|t\wedge {\gamma}^c|/2}$ where $v^{'} >0$ is a constant.
\end{lemma}
\begin{proof}
     Because nonzero eigenvalues of $X^{\top}_{\gamma} X_{\gamma}$ and $X_{\gamma} X_{\gamma}^\top$ are identical, it follows that $ Q_{\gamma} = \lambda^{|{\gamma}|/2}|X^{\top}_{\gamma} X_{\gamma}+\lambda I|^{-1/2} = |I+\lambda^{-1}X_{\gamma} X^{\top}_{\gamma}|^{-1/2}$. We first show that $ \frac{Q_{\gamma}}{Q_{\gamma \wedge t }} \leq (n\eta_m^n(\nu)/\lambda)^{-(r^*_\gamma - r_{\gamma\wedge t})/2}.$ There are two cases depending on $|\gamma| \leq,$ or $> u_n(\nu).$
 
 \emph{Case 1:} Suppose $|\gamma| \leq u_n(\nu).$ We then have
    \begin{align*}
      \frac{Q_{\gamma}}{Q_{\gamma \wedge t }} &= |I+\lambda^{-1}X_{\gamma} X^{\top}_{\gamma}|^{-1/2} |I+\lambda^{-1}X_{\gamma \wedge t } X^{\top}_{\gamma \wedge t }|^{1/2} \\ 
                                              & = \left|I+\lambda^{-1}X_{\gamma \wedge t } X^{\top}_{\gamma \wedge t } +  \lambda^{-1}X_{\gamma \wedge t^c }X_{\gamma \wedge t^c }^\top \right|^{-1/2} \left|I+\lambda^{-1}X_{\gamma \wedge t } X^{\top}_{\gamma \wedge t } \right|^{1/2}\\
                                              & = \left|I +  \lambda^{-1}X_{\gamma \wedge t^c }^\top\left(I+\lambda^{-1}X_{\gamma \wedge t } X^{\top}_{\gamma \wedge t } \right)^{-1}X_{\gamma \wedge t^c}\right|^{-1/2}.
    \end{align*}
    
    Next, using Sherman–Morrison–Woodbury matrix identity we have,
    \[\left(I+\lambda^{-1}X_{\gamma \wedge t } X^{\top}_{\gamma \wedge t } \right)^{-1} = I - X_{\gamma \wedge t }(X_{\gamma \wedge t }^\top X_{\gamma \wedge t } + \lambda I)^{-1}X_{\gamma \wedge t }^\top.\]
    
    Thus by letting $E=X_{\gamma \wedge t }^\top X_{\gamma \wedge t },$ $F = X_{\gamma \wedge t }^\top X_{\gamma \wedge t^c }$ and $G = X_{\gamma \wedge t^c }^\top X_{\gamma \wedge t^c }$ we have
    \begin{equation}\label{eq:Q-ratio1}
      \frac{Q_{\gamma}}{Q_{\gamma \wedge t }}
       = \left|\lambda^{-1}\{G+ \lambda I - F^\top(E + \lambda I)^{-1}F\} \right|^{-1/2}.
    \end{equation}

    However, note that $G+ \lambda I - F^\top(E + \lambda I)^{-1}F$ is the Schuar complement in
\[
    H = \begin{pmatrix}
     E + \lambda I & F \\ F^\top & G + \lambda I
    \end{pmatrix} = \begin{pmatrix}
     X_{\gamma \wedge t }^\top X_{\gamma \wedge t } + \lambda I & 
     X_{\gamma \wedge t }^\top X_{\gamma \wedge t^c } \\
     X_{\gamma \wedge t^c }^\top X_{\gamma \wedge t } & X_{\gamma \wedge t^c }^\top X_{\gamma \wedge t^c } + \lambda I,
    \end{pmatrix}
\]
so that the smallest eigenvalue of $G+ \lambda I - F^\top(E + \lambda I)^{-1}F$ is at least the smallest eigenvalue of $H$ which is, in turn, at least $n\eta_m^n(\nu) + \lambda$ because $H$ can be obtained by applying one permutation on the rows and columns of $X_\gamma^\top X_\gamma +\lambda I.$
Consequently, from \eqref{eq:Q-ratio1}, we finally have
\[ \frac{Q_{\gamma}}{Q_{\gamma \wedge t }} \leq  \left( \lambda^{-1} (\lambda + n\eta_m^n(\nu)\right)^{-r_{\gamma\wedge t^c}/2} \leq (n\eta_m^n(\nu)/\lambda)^{-(r^*_\gamma - r_{\gamma\wedge t})/2}\]
because $|\gamma\wedge t^c| \geq r_{\gamma\wedge t^c} \geq r_{\gamma} - r_{\gamma \wedge t} \geq   r_{\gamma}^* - r_{\gamma \wedge t}.$

\emph{Case 2:} If $|\gamma|>u_n(\nu)$ write $\gamma = \gamma'\vee \gamma''$ where $\gamma'$ and $\gamma''$ are disjoint, $|\gamma'| \leq u_n(\nu)$ and $\gamma'\wedge t = \gamma\wedge t.$ Then $Q_{\gamma\wedge t} = Q_{\gamma'\wedge t}$ and
\[Q_\gamma = |I + \lambda^{-1}X_{\gamma'}X_{\gamma'}^\top +
  \lambda^{-1}X_{\gamma''}X_{\gamma''}^\top|^{-1/2} \leq |I +
  \lambda^{-1}X_{\gamma'}X_{\gamma'}^\top|^{-1/2} = Q_{\gamma'}.\]
 Since $\gamma = \gamma'\vee \gamma''$,
$r_{\gamma} \ge r_{\gamma'}$ implying
$r^*_{\gamma} \ge r^*_{\gamma'}$. Also, $\gamma'\wedge t = \gamma\wedge t.$ Hence,
\[\frac{Q_\gamma}{Q_{\gamma \wedge t}} \leq \frac{Q_{\gamma'}}{Q_{\gamma'\wedge t}} \leq (n\eta_m^n(\nu)/\lambda)^{-(r^*_{\gamma'} - r_{\gamma'\wedge t})/2} \leq (n\eta_m^n(\nu)/\lambda)^{-(r^*_\gamma - r_{\gamma\wedge t})/2}. \]
Furthermore,
    \begin{align*}
        \frac{Q_{\gamma \wedge t }}{Q_t} &= |I+\lambda^{-1}X_{\gamma \wedge t } X^{\top}_{\gamma \wedge t }|^{-1/2} |I+\lambda^{-1}X_t X^{\top}_t|^{1/2}\\
        &= \left|\left(I+\lambda^{-1}X_{\gamma \wedge t } X^{\top}_{\gamma \wedge t }\right)^{-1} \left(I+\lambda^{-1}X_{\gamma \wedge t } X^{\top}_{\gamma \wedge t } + \lambda^{-1}X_{\gamma^c \wedge t } X^{\top}_{\gamma^c \wedge t } \right)\right|^{1/2}\\
        &= \left|I + \lambda^{-1}X_{\gamma^c \wedge t }\left(I+\lambda^{-1}X_{\gamma \wedge t } X^{\top}_{\gamma \wedge t }\right)^{-1} X^{\top}_{\gamma^c \wedge t }\right|^{1/2} \\
        &\leq \left|I+\lambda^{-1}X_{\gamma^c \wedge t } X^{\top}_{\gamma^c \wedge t }\right|^{1/2} \leq v'\left(n/\lambda\right)^{|{\gamma}^c \wedge t |/2},
    \end{align*}
    for some $v'>0,$ where the second to last inequality holds because
    $I+\lambda^{-1}X_{\gamma \wedge t } X^{\top}_{\gamma \wedge t} \ge
    I$ and the last inequality holds since the condition
    C\ref{c.tc_bound} is in force and the fact that $X_{\gamma^c \wedge t}$ is a
    submatrix of $X_t.$ Since $X_t$ is full rank, $r_t = r_{\gamma \wedge t} + r_{\gamma^c \wedge t} = r_{\gamma \wedge t} + |\gamma^c \wedge t|$. Thus finally we have,
    \begin{align*}
        \frac{Q_{\gamma}}{Q_t} & \leq v' \left(n\eta_m^n(\nu)/\lambda\right)^{-(r^{*}_{\gamma}/2)} \left(n\eta_m^n(\nu)/\lambda\right)^{(r_t -|\gamma^c \wedge t|)/2} \left(n/\lambda\right)^{|{\gamma}^c \wedge t |/2}\\
        &= v' (n\eta_m^n(\nu)/\lambda)^{-(r^{*}_{\gamma} - r_t)/2} (\eta_m^n(\nu))^{-|\gamma^c \wedge t|/2}.
    \end{align*}

\end{proof}
\noindent Then, using Lemma \ref{lm.Qgamma}, we have the following corollary. 
\begin{corollary}\label{cl.post_ratio}
    For any model $\gamma \not = t$, 
    \begin{equation*}
        \begin{aligned} 
        \operatorname{PR}({\gamma}, t) = \frac{f(\gamma|Y, \sigma^2)}{f(t|Y, \sigma^2)} 
        \leq & v^{'}\left(n \eta_m^n(\nu)/\lambda \right)^{-\left(r_{\gamma}^{*}-r_{t}\right)/2}\left(\eta_m^n(\nu)\right)^{-\left|\gamma^c \wedge t \right|/2} b_n^{|{\gamma}|-|t|} \\ & \times \exp \left\{-\frac{1}{2 \sigma^{2}}\left(R_{\gamma}-R_{t}\right)\right\},
        \end{aligned}
    \end{equation*}
    where $b_n=w/(1-w)\sim p^{-1}$, and $v^{'}>0$ is a constant.
\end{corollary}
\begin{proof}
The posterior of the model $\gamma$ under (2a)-(2d) is given by 
   \begin{align*}
        f(\gamma |Y, \sigma^2) \propto &\exp\left\{-\frac{1}{2\sigma^2}\left(\tilde{Y}^{\top} \tilde{Y}-\tilde{Y}^{\top} X_{\gamma}\left(X^{\top}_{\gamma} X_{\gamma}+\lambda I\right)^{-1} X_{\gamma}^{\top} \tilde{Y} \right)\right\}\\ & \ \ \ \ \ \times \lambda^{|\gamma| / 2}\left|X_{\gamma}^{\top} X_{\gamma}+\lambda I\right|^{-1 / 2} w^{|\gamma|}(1-w)^{p-|\gamma|}\\
        \leq & Q_{\gamma} b_n^{|\gamma|} \exp \left\{-\frac{1}{2 \sigma^{2}}R_{\gamma}\right\}.
   \end{align*}
   
Hence from lemma~\ref{lm.Qgamma} we have
    \begin{equation*}
        \begin{aligned} 
        \operatorname{PR}({\gamma}, t) = \frac{f(\gamma|Y, \sigma^2)}{f(t|Y, \sigma^2)}
        =& \frac{Q_{\gamma}}{Q_{t}} b_n^{|\gamma|-|t|} \exp \left\{-\frac{1}{2 \sigma^{2}}\left(R_{\gamma}-R_{t}\right)\right\} \\ 
        \leq & v^{'}\left(n\eta_m^n(\nu)/\lambda\right)^{-\left(r_{\gamma}^{*}-r_{t}\right)/2}\left(\eta_m^n(\nu)\right)^{-\left|{\gamma}^c \wedge t \right|/2} b_n^{|{\gamma}|-|t|} \\ & \times \exp \left\{-\frac{1}{2 \sigma^{2}}\left(R_{\gamma}-R_{t}\right)\right\}. 
        \end{aligned}
    \end{equation*}

\end{proof}

\begin{lemma}\label{lm.tail_prob}
    For any sequence $h_n \rightarrow \infty$, we have\\
    \centerline{$P(R_t - R^{*}_t > h_n) \leq \exp(-c'n h_n/\lambda)$ for some $c'>0$.}
\end{lemma}
\begin{proof}
  Since
  $(n/\lambda)I + (\frac{X^{\top}_t X_t}{n})^{-1} \geq (n/\lambda) I$
  and $1_n^{\top}X =0$, Sherman–Morrison–Woodbury identity implies
    \begin{align*}
        0 \leq R_t-R_t^* &=Y^{\top} X_t\left[(X^{\top}_t X_t)^{-1}-\left(\lambda I + X^{\top}_t X_t\right)^{-1}\right]X^{\top}_t Y\\
                    &=Y^{\top} X_t(X^{\top}_t X_t)^{-1} \left(\lambda^{-1} I + (X^{\top}_t X_t)^{-1}\right)^{-1}(X^{\top}_t X_t)^{-1} X^{\top}_t Y\\
                    &\leq \left(n/\lambda\right)^{-1}Y^{\top}WY,
    \end{align*}
    where $W=nX_t (X^{\top}_t X_t)^{-2}X^{\top}_t$ has rank $|t|$ and by condition C5 has bounded eigenvalues. We want to show that $P(R_t - R^{*}_t > h_n) \leq P\left(Y^{\top}WY > n\lambda^{-1} h_n\right) \leq\exp\left(-c'n\lambda^{-1} h_n\right)$. Next, since $1_n^{\top}X =0$, we have  
    \begin{align*}
        Y^{\top}WY =& (\beta^{\top}_t X^{\top}_t + \beta^{\top}_{t^c} X^{\top}_{t^c} + \epsilon^{\top}) W (X_t\beta_t + X_{t^c}\beta_{t^c} + \epsilon)\\
             =& \beta^{\top}_t X^{\top}_t W X_t\beta_t + \beta^{\top}_t X^{\top}_t W X_{t^c}\beta_{t^c} + \beta^{\top}_t X^{\top}_t W \epsilon + \beta^{\top}_{t^c} X^{\top}_{t^c} W X_t\beta_t \\
             &+ \beta^{\top}_{t^c} X^{\top}_{t^c} W X_{t^c}\beta_{t^c} + \beta^{\top}_{t^c} X^{\top}_{t^c} W \epsilon + \epsilon^{\top}W X_t\beta_t + \epsilon^{\top}W X_{t^c}\beta_{t^c} + \epsilon^{\top}W\epsilon\\
             =& n\beta^{\top}_t\beta_t + \beta^{\top}_t X^{\top}_t W X_{t^c}\beta_{t^c} + \beta^{\top}_{t^c} X^{\top}_{t^c} W X_t\beta_t + \beta^{\top}_t X^{\top}_t W \epsilon + \epsilon^{\top}W X_t\beta_t\\ 
             &+ \beta^{\top}_{t^c} X^{\top}_{t^c} W \epsilon + \epsilon^{\top}W X_{t^c}\beta_{t^c} + \beta^{\top}_{t^c} X^{\top}_{t^c} W X_{t^c}\beta_{t^c} + \epsilon^{\top}W\epsilon.
    \end{align*}
    To find the bound of the tail probability, we spilt our proof into following five steps: 
    \begin{enumerate}[label=(\roman*)]
        \item First, we want to show that $|\beta^{\top}_t X^{\top}_t W X_{t^c}\beta_{t^c}| \preceq \sqrt{n\log p}$. It is clear that
        \begin{equation*}
            |\beta^{\top}_t X^{\top}_t W X_{t^c}\beta_{t^c}| \leq \|X_t\beta_t\| \|X_{t^c}\beta_{t^c}\| \alpha_{max}(W),
        \end{equation*}
        and $\|X_t\beta_t\|^2=\beta^{\top}_t X^{\top}_t X_t\beta_t=n\beta^{\top}_t\left(\frac{X^{\top}_t X_t}{n}\right) \beta_t \leq n c_1 \|\beta_t\|^2$ for some constant $c_1 > 0$. By condition C5 we also know that $\alpha_{max}(W)$ is bounded.
        Then with $\|X_{t^c}\beta_{t^c}\| \preceq \sqrt{\log p}$ from condition C\ref{c.tc_bound}, we have $|\beta^{\top}_t X^{\top}_t W X_{t^c}\beta_{t^c}| \leq c_2 \sqrt{n\log p} $ for some constant $c_2 > 0$.
        
        \item 
        Next we will show that $|\beta^{\top}_{t^c} X^{\top}_{t^c} W X_{t^c}\beta_{t^c}| \preceq \log p$. Condition C\ref{c.tc_bound} and the fact that $W$ has bounded eigenvalues implies that
        \begin{equation*}
            |\beta^{\top}_{t^c} X^{\top}_{t^c} W X_{t^c}\beta_{t^c}| \leq \| X_{t^c}\beta_{t^c} \|^2 \alpha_{max}(W) \leq c_3 \log p
        \end{equation*}
        
 \item Next, we will show that
        $P(\epsilon^{\top} W \epsilon \geq a) \leq P(c_4 \chi^2_{1,|t|}
        \geq a)$ for all $a > 0,$ and $n \geq 1,$ and for some $c_4>0$, where $\chi^2_{1,|t|}$ is  distributed as $\chi^2$ with $|t|$ degrees of freedom.  To that end, let
        $W=P\Lambda P^{\top}$, where $P$ is an orthogonal matrix and
        $\Lambda=\text{diag}\{\lambda_1, \ldots, \lambda_n\}$ is the
        diagonal matrix of the eigenvalues of $W$. Then, since
        $P^{\top}\epsilon \sim N(0, \sigma^2 I)$,
        $\epsilon^{\top} W \epsilon/\sigma^2 = \epsilon^{\top} P
        \Lambda P^{\top}\epsilon/\sigma^2 = \sum_{i=1}^{|t|} \lambda_i
        \frac{G_i^2}{\sigma^2},$ where
        $G_i \stackrel{iid}{\sim} N(0, \sigma^2), i=1, \ldots,
        |t|$. Since eigenvalues of $W$ are bounded,
        $\lambda_i \sigma^2 \leq c_5$ for some $c_4>0$ for all
        $i=1, \ldots, |t|$. Hence,
        $P(\epsilon^{\top} W \epsilon \geq a) = P(\sigma^2
        \sum_{i=1}^{|t|} \lambda_i \frac{G_i^2}{\sigma^2} \geq a) \leq
        P(c_4 \chi^2_{1,|t|} \geq a)$.
        
        \item
        Next we want to show $P(| \beta^{\top}_{t^c} X^{\top}_{t^c} W \epsilon | \geq a) \leq P(c_5\chi_{2,|t|} \geq a/\sqrt{\log p})$ for all $a > 0$, $n \geq 1$ and for some $c_5 > 0,$ where $\chi^2_{2,|t|}$ is distributed as $\chi^2$ with $|t|$ degrees of freedom. To that end, note that by Cauchy-Schwarz inequality,
        \[ | \beta^{\top}_{t^c} X^{\top}_{t^c} W \epsilon | \leq \| \beta^{\top}_{t^c} X^{\top}_{t^c}\|\| W \epsilon \|\]
        However, as in step (iii) above, $\|W\epsilon\|^2$ is stochastically dominated by a constant multiple of $\chi^2-$distributed random variable since $W^2 = P\Lambda^2 P^\top$ has rank $|t|$ and bounded eigenvalues. Hence by Condition C\ref{c.tc_bound}, there exists $c_5 > 0,$ such that
        \[P(| \beta^{\top}_{t^c} X^{\top}_{t^c} W \epsilon | \geq a) \leq P(\| \beta^{\top}_{t^c} X^{\top}_{t^c}\|\| W \epsilon \| \geq a) \leq P\bigl(c_5(\log p)^{1/2}\chi_{2,|t|} \geq a\bigr)\]
      
        \item 
        Thus all sufficiently large $n$, we have
   \[Y^\top WY \leq 2\beta^{\top}_t X^{\top}_t W \epsilon + 2\beta_{t^c}X_{t^c}W\epsilon + \epsilon^{\top} W \epsilon + n \beta^{\top}_t \beta_t + 2c_2 \sqrt{n\log p}  + c_3 \log p.\]
        Now note that  $2\beta^{\top}_t X^{\top}_t W \epsilon = 2n \beta^{\top}_t (X^{\top}_t X_t)^{-1} X^{\top}_t \epsilon \sim N(0, nC_n^2)$, where 
        $$C_n^2 = 4\sigma^2 \beta^{\top}_t \left(\frac{X^{\top}_t X_t}{n}\right)^{-1} \beta_t$$ 
   is bounded. Also note that 
   $$n\lambda^{-1}h_n -  n \beta^{\top}_t \beta_t - 2c_2 \sqrt{n\log p} - c_3\log p > n\lambda^{-1}h_n/2$$ and $n\lambda^{-1}h_n/(\log p) > 1$ for sufficiently large $n.$
   Thus, for sufficiently large $n,$ we have,
   \begin{align*}
&\ P\left(Y^{\top}WY > n\lambda^{-1} h_n\right) \\ 
& \leq P\left(2\beta^{\top}_t X^{\top}_t W \epsilon + 2\beta_{t^c}X_{t^c}W\epsilon + \epsilon^{\top} W \epsilon > \frac{1}{2}n\lambda^{-1}h_n\right) \\
& \leq P\left(2\beta^{\top}_t X^{\top}_t W \epsilon > \frac{1}{6}n\lambda^{-1}h_n\right) + P\left(2\beta_{t^c}X_{t^c}W\epsilon  > \frac{1}{6}n\lambda^{-1}h_n\right) + P\left( \epsilon^{\top} W \epsilon > \frac{1}{6}n\lambda^{-1}h_n\right)  \\
& \leq \ P\left(\frac{2\beta^{\top}_t X^{\top}_t W \epsilon}{\sqrt{n}C_n} > \frac{\sqrt{n} }{6} \frac{h_n}{C_n\lambda} \right) +  P\left(c_5^2\chi^2_{2,|t|} > \frac{1}{36}\frac{\left(n\lambda^{-1} h_n\right)^2}{\log p}\right) + P\left(c_4\chi^2_{1,|t|} > \frac{1}{6}\frac{n h_n}{\lambda}\right) \\
& \leq  \exp\left(-c''n\lambda^{-2} h_n^2\right) + 
           \exp\left(-c'''n\lambda^{-1} h_n\right) + 
           \exp\left(-c''''n\lambda^{-1} h_n\right) \\
& \leq  \exp\left(-c'n\lambda^{-1} h_n\right),
   \end{align*}
for some positive constants  $c',c'',c'''$ and, $c''''$. 
    \end{enumerate}

\end{proof}

\section{Proof of Theorem \ref{thm.consistency}}\label{sec:proofKnownSigma}
To prove the model selection consistency, we use the same strategy as in \citet{nari:he:2014} by dividing the set of models into the following subsets:
\begin{enumerate}[label=(\roman*)]
\item Unrealistically large models:
  $M_{1}=\left\{{\gamma} : r_{\gamma}>u_{n}\right\}$, the models of
  rank greater than $u_n$. Abusing notation we use $u_n$ and 
    $u_n(\nu)$ interchangeably.
  \item Over-fitted models:
    $M_{2}=\left\{{\gamma} : {\gamma} \supset t, r_{\gamma} \leq
      u_{n}\right\}$, the models of rank smaller than $u_n$ which
    include all the important variables and at least one unimportant
    variables.
  \item Large models:
    $M_3=\left\{{\gamma}:{\gamma} \not \supset t, J|t|<r_{\gamma} \leq
      u_n \right\}$, that is, the models which miss one or more
    important variables with rank greater than $J|t|$ but smaller
    than $u_n$ for some fixed positive integer $J$.
  \item Under-fitted models:
    $M_4=\{{\gamma}:{\gamma} \not \supset t, r_{\gamma} \leq J|t|\}$,
    the models which have rank smaller than $J|t|$ and miss at
    least one important variable.
\end{enumerate}
We aim to show that $\sum_{\gamma\in M_k} \operatorname{PR}({\gamma}, t)\xrightarrow{P} 0$ for each $k = 1,2,3,4$, with $\sigma^2$ known.

\subsection{Unrealistically large models}\label{sec:proofverylarge}
We first want to prove that the sum of posterior ratios $PR(\gamma,t)$ over $\gamma \in M_1$ converges exponentially to zero. Note that $M_1$ is empty if $p < n/\log p^{2+\nu}$. The reason is that if $p < n/\log p^{2+\nu}$, then $u_n(\nu) = p$ and $r_\gamma \leq p$, which contradicts the definition of $M_1$. First, we want to find a set of events that are almost unlikely to happen. So, note that for any $s>0$,
\begin{align}
  \label{eq:ularger}
    &\ \ \ \ P\left[\cup_{\gamma \in M_{1}}\left\{R_{t}-R_{\gamma}>n(1+2 s) \sigma^{2}\right\}\right] \nonumber\\ 
    &\leq P\left[R_{t}>n(1+2 s) \sigma^{2}\right] \nonumber\\ &\leq P\left[R_{t}^{*}>(1+s)n\sigma^{2}\right]+P\left[R_{t}-R_{t}^{*}>s n \sigma^{2}\right] \nonumber\\
    &= P\left[\frac{R_t^{*}}{n \sigma^2} -1 > s\right] + P\left[R_{t}-R_{t}^{*}>s n \sigma^{2}\right] \nonumber\\ 
    &\leq \exp\{-c n\} + \exp\{-c' s n^2 \sigma^2/\lambda \},
\end{align}
for some $c, c'>0$, due to Lemma A.2 of \citet{nari:he:2014}.

\noindent Now we consider the term $n\eta_m^n(\nu)/\lambda$. Condition
C\ref{c.shrink_rates} indicates that
$n\eta_m^n(\nu)/\lambda \preceq n \vee p^{2+3\delta}$ because
$\eta_m^n(\nu)$ is the smallest eigenvalue of a correlation matrix,
i.e., $\eta_m^n(\nu) < 1,$ and condition C\ref{c.eigen_g} implies that
$n\eta_m^n(\nu)/\lambda \succeq n \vee p^{2+2\delta}$, that is,
\begin{equation} \label{eq:eigenbound}
    (n \vee p^{2+2\delta}) \preceq n\eta_m^n(\nu)/\lambda \preceq (n \vee p^{2+3\delta}).
\end{equation}
Then we restrict our attention to the high probability event
$\cap_{\gamma \in M_{1}} \left\{R_{t}-R_{\gamma} \leq n(1+2 s)
  \sigma^{2}\right\}$ for $s<\delta / 2(2+\delta)$. Note that, in this
case, the upper bound \eqref{eq:ularger} of the probability of the
complement of this event is bounded by $2 \exp\{-c'' n\}$ for some $c'' >0.$ First, by Lemma
\ref{lm.Qgamma}, we have
\begin{align*} 
    \sum_{\gamma \in M_{1}} \operatorname{PR}({\gamma}, t)
    & \preceq \sum_{\gamma \in M_{1}} v' \left(n\eta_m^n(\nu)/\lambda\right)^{-(r_{\gamma}^{*}-r_t)/2} \left(\eta_m^n(\nu)\right)^{-|\gamma^c \wedge t| / 2} b_n^{|\gamma|-|t|} e^{n(1+2 s) / 2} \\ 
    & \preceq \sum_{\gamma \in M_{1}} p^{-(1+\delta)\left(u_{n}-|t|\right)} \left(\eta_m^n(\nu)\right)^{-|t| / 2} b_n^{|\gamma|-|t|} e^{n(1+2 s) / 2}.
\end{align*}
because for all $\gamma \in M_1$,
$r^*_\gamma = r_\gamma \wedge u_n = u_n$,
$|\gamma^c \wedge t| \le |t|$ and condition C\ref{c.eigen_g} is in
force. Recall that $b_n \sim p^{-1}$. Thus, $(1+b_n)^{p} \sim 1$. Also, by condition C1,
$p = \exp(n d_n)$ for some $d_n \rightarrow 0$. Then, due to condition
C\ref{c.eigen_g} $u_n = n/\log p^{2+\nu} \geq n/\log p^{2+\delta}$
since $\nu < \delta$, we have
\begin{align*}
    \sum_{\gamma \in M_{1}} \operatorname{PR}({\gamma}, t) 
    & \preceq \sum_{\gamma \in M_{1}} e^{-(1+\delta)(u_n-|t|)\text{log} p} b_n^{|{\gamma}|-|t|} (\eta_m^n(\nu))^{-|t|/2}e^{n(1+2 s)/2}\\
    & \preceq \sum_{\gamma \in M_{1}} e^{-(1+\delta)\frac{n}{(2+\delta)\text{log} p} \text{log} p} e^{n(1+2 s)/2} p^{\kappa|t|/2} b_n^{|{\gamma}|-|t|}\\
    & \preceq e^{-n(1+\delta) /(2+\delta)} e^{n(1+2 s) / 2} p^{\kappa|t|/2} \sum_{\gamma \in M_{1}} b_n^{(|{\gamma}|-|t|)} \\ & \preceq e^{-n(1+\delta) /(2+\delta)} e^{n(1+2 s) / 2} p^{(1+\kappa/2)|t|} \sum_{|{\gamma}|=u_{n}}^{p}\left(\begin{array}{c}{p} \\ {|{\gamma}|}\end{array}\right) b_n^{|{\gamma}|} \\ & \preceq e^{-n(1+\delta) /(2+\delta)} e^{n(1+2 s) / 2} e^{n(1+\kappa/2)|t|d_n}\left(1+b_n\right)^{p} \\ 
    & \preceq e^{-v' n} \longrightarrow 0, 
\end{align*} 
as $n \rightarrow \infty$ for some $v' >0$, if $s$ satisfies $1+2 s<2(1+\delta) /(2+\delta)$, i.e., $s<\delta / 2(2+\delta)$. Therefore, we have 
\begin{equation}\label{eq.verylarge}
    \sum_{\gamma\in M_1} \operatorname{PR}({\gamma}, t)\xrightarrow{P} 0.
\end{equation}

\subsection{Over-fitted models}\label{sec:proofoverfitted}
Models in $M_2$ include all important variables plus one or more unimportant variables. For $\gamma \in M_2$, 
\begin{align*} 
    R_{t}^{*}-R_{\gamma}^{*} &= Y^{\top}(I-P_t)Y - Y^{\top}(I-P_{\gamma})Y = \Vert(P_{\gamma}-P_t)(X_t\beta_t + X_{t^c}\beta_{t^c}+\epsilon)\Vert^2 \\
    &= \left(\left\|\left(P_{\gamma}-P_{t}\right) X_{t^c} \beta_{t^c}\right\|+\left\|\left(P_{\gamma}-P_{t}\right) \epsilon\right\|\right)^{2}  \leq\left(\left\|X_{t^c} \beta_{t^c}\right\|+\sqrt{\epsilon^{\top} (P_{\gamma} -P_{t}) \epsilon}\right)^{2}.
\end{align*}
Due to Lemma 1 of \citet{laur:mass:2000} and the fact that
$\epsilon^{\top}(P_{\gamma} - P_t)\epsilon/\sigma^2 \sim \chi^2_{r_{\gamma} -
  r_t}$, for any $x>0$ and for some $\sqrt{2/3} < v < 1$, we
  have for all sufficiently large $n$,
\begin{equation}
\begin{aligned}\label{eq:rtbound_overfitted}
    &\ \ \ \ P\left[R_{t}^{*}-R_{\gamma}^{*}>\sigma^{2}(2+3 x)\left(r_{\gamma}-r_{t}\right) \text{log} p\right]\\ 
    & \leq P\left[\left(\left\|X_{t^c} \beta_{t^c}\right\|+\sqrt{\epsilon^{\top} (P_{\gamma} -P_{t}) \epsilon}\right)^{2} > \sigma^{2}(2+3 x)\left(r_{\gamma}-r_{t}\right) \log p\right]\\
    &= P\left[\frac{\epsilon^{\top} (P_{\gamma} -P_{t}) \epsilon}{\sigma^{2}(2+3x)\left(r_{\gamma}-r_{t}\right) \log p} > 1-\frac{2\left\|X_{t^c} \beta_{t^c}\right\| \sqrt{\sigma^{2}(2+3x)\left(r_{\gamma}-r_{t}\right) \log p}-\left\|X_{t^c} \beta_{t^c}\right\|^2}{\sigma^{2}(2+3x)\left(r_{\gamma}-r_{t}\right) \log p}\right]\\
    &\leq P\left[\epsilon^{\top} (P_{\gamma} -P_{t}) \epsilon>\frac{1}{2}\sigma^{2}(2+3 v x)\left(r_{\gamma}-r_{t}\right) \log p\right]\\ 
    &\leq P\left[\chi_{r_{\gamma}-r_{t}}^{2}-\left(r_{\gamma}-r_{t}\right)>\frac{1}{2}\{(2+3 v x) \log p - 1\} \left(r_{\gamma}-r_{t}\right) \right] \\  
    &\leq P\left[\chi_{r_{\gamma}-r_{t}}^{2}-\left(r_{\gamma}-r_{t}\right)>\frac{1}{2}\left(2+3 v^{2} x\right)\left(r_{\gamma}-r_{t}\right) \log p\right] \\ 
    &\leq P\Big\{ \chi_{r_{\gamma}-r_{t}}^{2}-\left(r_{\gamma}-r_{t}\right) > \sqrt{(r_{\gamma} - r_t) \left[(r_{\gamma}-r_t)(1+x)\log p + a\right]} + \left[(r_{\gamma}-r_t)(1+x)\log p + a\right]\Big\}\\ 
    &\leq \exp\left\{ -(r_{\gamma}-r_t)(1+x)\log p + a\right\}\\ 
    &\leq c_1 \exp \left\{-(1+x)\left(r_{\gamma}-r_{t}\right) \log p\right\} =c_1 p^{-(1+x)\left(r_{\gamma}-r_{t}\right)}, 
\end{aligned}
\end{equation}
where $c_1 = \exp(a) > 0$ and $a$ is a constant such that
\begin{align*}
    &\ \ \ \ \sqrt{(r_{\gamma} - r_t) \left[(r_{\gamma}-r_t)(1+x)\log p + a\right]} + [(r_{\gamma}-r_t)(1+x)\log p + a]
    < \left(2+3 v^{2} x\right)\left(r_{\gamma}-r_{t}\right) \log p.
\end{align*}
Now, consider $0<s<\delta/8$ and define the event 
\begin{align*}
    E_1({\gamma}) &:= \left\{R_t-R_{\gamma} > 2\sigma^2(1+4s)(r_{\gamma} -r_t)\log p\right\}
    \subset \left\{R_t-R_{\gamma} > 2\sigma^2(1+2s)(r_{\gamma} -r_t)\log p\right\}.
\end{align*}
Then, for a fixed dimension $d>r_t$, consider the event $U(d):=\bigcup_{\left\{{\gamma}: r_{\gamma}=d\right\}} E_1({\gamma})$. Since $R_{\gamma} \geq R_{\gamma}^{*}$, we have 
\begin{align*}
   P[U(d)] \leq & P\left[\cup_{\left\{{\gamma} : r_{\gamma}=d\right\}}\left\{R_{t}-R_{\gamma}>2 \sigma^{2}(1+2 s)\left(r_{\gamma}-r_{t}\right) \log p\right\}\right] \\
   \leq & P\left[\cup_{\left\{{\gamma} : r_{\gamma}=d\right\}}\left\{R_{t}-R_{\gamma}^{*}>2 \sigma^{2}(1+2 s)\left(r_{\gamma}-r_{t}\right) \log p\right\}\right] \\ 
   \leq & P\left[\cup_{\left\{{\gamma} : r_{\gamma}=d\right\}}\left\{R_{t}^{*}-R_{\gamma}^{*}>\sigma^{2}(2+3 s)\left(d-r_{t}\right) \log p\right\}\right] \\ &+P\left[R_{t}-R_{t}^{*}>s \sigma^{2}\left(d-r_{t}\right) \log p\right]\\
   \leq& \sum_{\gamma:r_{\gamma}=d} P\left[R_t^{*} - R_{\gamma}^{*} > \sigma^{2}(2+3 s)\left(d-r_{t}\right) \log p\right]\\
   &+P\left[R_{t}-R_{t}^{*}>s \sigma^{2}\left(d-r_{t}\right) \log p\right]\\
   \leq& \sum_{\gamma:r_{\gamma}=d} c_1p^{-(1+s)(d-r_t)}+\exp\left\{-c'ns\sigma^2(d-r_t)(\log p)/\lambda \right\}\\
   \leq& c_1p^{-(1+s)(d-r_t)}p^{d-r_t}+\exp\left\{-c' s(d-r_t)\log p\right\}\\
   =& c_1p^{-s(d-r_t)} + p^{-c' s(d-r_t)}\\
   \leq& c_3 p^{-c_4s(d-r_t)},
\end{align*}
for some $c_3, c_4 >0$, where the fifth and the sixth inequality hold
due to (\ref{eq:rtbound_overfitted}), Lemma \ref{lm.tail_prob},
condition C2, and the fact that the event
$\left\{R_t^{*} - R_{\gamma}^{*} > \sigma^{2}(2+3
  s)\left(d-r_{t}\right) \log p\right\}$ depends only on the
projection matrix $P_{\gamma \wedge t^c}$, so we can write the union
$\cup_{\left\{{\gamma} : r_{\gamma}=d\right\}}$ as a smaller set of
events indexed by $P_{\gamma \wedge t^c}$. Note that since there
exists at most $p^k$ subspaces of rank $k$, the cardinality of such
projections is at most $p^{d-r_t}$.  Next, we consider the union of
all such events $U(d)$, that is, \begin{equation*}
    \begin{aligned} P\left[\cup_{\left\{d>r_{t}\right\}} U(d)\right] & \leq \sum_{\left\{d>r_{t}\right\}} P[U(d)] \leq c_3 \sum_{d>r_{t}} p^{-c_4s\left(d-r_{t}\right)} \\ 
    & \leq  c_3 \sum_{d-r_t=1}^{\infty} p^{-c_4s\left(d-r_{t}\right)} \leq c_3 \frac{p^{-c_4s}}{1-p^{-c_4s}} \\ &=\frac{c_3}{p^{c_4s}-1} \longrightarrow 0 \text { as } n \rightarrow \infty \end{aligned}.
\end{equation*}
Note that, $r_{\gamma}^{*}=r_{\gamma}$ as $r_{\gamma} < u_n$ for
${\gamma} \in M_2$. Then again restricting to the high probability
event $\cap_{\left\{d>r_{t}\right\}} U(d)^c$, by Lemma
\ref{lm.Qgamma}, \eqref{eq:eigenbound} and the fact that
$\gamma^c \wedge t$ is empty, we have 
\begin{equation*}
    \begin{aligned} 
    \sum_{\gamma \in M_{2}} \operatorname{PR}({\gamma}, t) \preceq & \sum_{\gamma \in M_{2}}\left(n\eta_m^n(\nu)/\lambda\right)^{-\left(r_{\gamma}^{*}-r_{t}\right)/2} b_n^{(|{\gamma}|-|t|)} \left(\eta_m^n(\nu)\right)^{-|\gamma^c \wedge t|/2}\\ &\times \exp \left\{-\frac{1}{2 \sigma^{2}}\left(R_{\gamma}-R_{t}\right)\right\} \\ 
    \preceq & \sum_{\gamma \in M_{2}}\left(p^{-(2+2\delta)} \wedge n^{-1}\right)^{\left(r_{\gamma}-r_{t}\right)/2} b_n^{(|{\gamma}|-|t|)} p^{(1+4 s)\left(r_{\gamma}-r_{t}\right)} \\
    \preceq & \sum_{\gamma \in M_{2}}\left(p^{1+\delta} \vee \sqrt{n}\right)^{-\left(r_{\gamma}-r_{t}\right)} b_n^{(|{\gamma}|-|t|)} p^{(1+4 s)\left(r_{\gamma}-r_{t}\right)} \\ 
    \preceq & \sum_{\gamma \in M_{2}}\left(p^{1+\delta-1-4s} \vee \sqrt{n}p^{-1-4s}\right)^{-\left(r_{\gamma}-r_{t}\right)} b_n^{(|{\gamma}|-|t|)}\\
     \preceq & \sum_{\gamma \in M_{2}}\left(p^{-\delta/2} \wedge \frac{p^{1+\delta/2}}{\sqrt{n}}\right)^{\left(r_{\gamma}-r_{t}\right)} b_n^{(|{\gamma}|-|t|)}\\ 
    \preceq & \left(p^{-\delta / 2} \wedge \frac{p^{1+\delta / 2}}{\sqrt{n}}\right)\sum_{|{\gamma}|=|t|+1}^{p}\left(\begin{array}{c}{p} \\ {|{\gamma}|-|t|}\end{array}\right) b_n^{(|{\gamma}|-|t|)}\\ 
    \preceq & \left(p^{-\delta / 2} \wedge \frac{p^{1+\delta / 2}}{\sqrt{n}}\right) (1+b_n)^{p}\\  
    \sim & \rho_{n} \longrightarrow 0,
    \end{aligned}
\end{equation*}
as $n\rightarrow \infty$, where
$\rho_{n}=\left(p^{-\delta / 2} \wedge \frac{p^{1+
      \delta/2}}{\sqrt{n}}\right)$. In the above, we used the fact that $\rho_n \leq 1$ and $r_\gamma - r_t \ge 1$. Note that,
$\delta-4s \geq \frac{\delta}{2} $ since $0 < s < \delta/8$. Hence, we
have
\begin{equation}\label{eq.overfitted}
    \sum_{\gamma\in M_2} \operatorname{PR}({\gamma}, t)\xrightarrow{P} 0.
\end{equation}

\subsection{Large Models}\label{sec:prooflarge}
For models in $M_3$ where the rank is at least $J|t|$ and one or more important variables are not inclued, similar to what we've shown in Section \ref{sec:proofoverfitted}, for $0< s < \delta/d$ we use the event
\begin{equation*}
    E_1({\gamma}) \subset \left\{R_t-R_{\gamma \vee t }> 2\sigma^2(1+2s)(r_{\gamma} -r_t)\log p\right\}.
\end{equation*}\\
Then we consider the union of such events $U(d)=\bigcup_{\left\{{\gamma}: r_{\gamma}=d\right\}} E_1({\gamma})$, for $d>J|t|$, and $s=\delta/8$. Using \eqref{eq:rtbound_overfitted} with the fact that $r_{\gamma \vee t} \ge r_\gamma$ and Lemma \ref{lm.tail_prob} we have
\begin{equation*}
    \begin{aligned} 
    P[U(d)] \leq & P\left[\cup_{\left\{{\gamma} : r_{\gamma}=d\right\}}\left\{R_{t}-R_{\gamma \vee t }>2 \sigma^{2}(1+2 s)\left(r_{\gamma}-r_{t}\right) \log p\right\}\right] \\ 
    \leq & P\left[\cup_{\left\{{\gamma} : r_{\gamma}=d\right\}}\left\{R_{t}-R_{\gamma \vee t }^{*}>2 \sigma^{2}(1+2 s)\left(r_{\gamma}-r_{t}\right) \log p\right\}\right] \\ 
    \leq & \sum_{\{{\gamma}:r_{\gamma} =d\}}P\left[R_t^{*}-R_{\gamma \vee t }^{*}>\sigma^2(2+3s)(d-r_t)\log p\right]\\ &\ \ + P[R_t-R_t^{*} > s\sigma^2(d-r_t)\log p]\\
    \leq & c_1 p^{-(1+s)\left(d-r_{t}\right)} p^{d} + e^{-c'ns\sigma^2(d-r_t)(\log p)/\lambda }\\ 
    \leq & c_5p^{-c_6 d},
    \end{aligned}
  \end{equation*}
  for some $c_5, c_6 >0$.

\noindent Then,   
\begin{equation*}
    P[\cup_{\{d>J|t|\}} U(d)] \leq \sum_{d>J|t|}P\left[U(d)\right] \leq \sum_{d>J|t|} c_5 p^{-c_6 d} \longrightarrow 0 \text{ as } n\rightarrow \infty.
\end{equation*}

\noindent Now, we restrict our attention to the high probability event $\bigcap_{\left\{d>r_{t}\right\}} U(d)^{c}$, we have
\begin{align*}
    \sum_{\gamma \in M_{3}} \operatorname{PR}({\gamma}, t)
    \preceq& \sum_{\gamma \in M_{3}} \left(n\eta_m^n(\nu)/\lambda\right)^{-(r_{\gamma}^{*}-r_t)/2}\left(\eta_m^n(\nu)\right)^{-|\gamma^c \wedge t|/2} b_n^{(|{\gamma}|-|t|)}\\
    &\times\exp \left\{-\frac{1}{2 \sigma^{2}}\left(R_{\gamma}-R_{t}\right)\right\}\\
    \preceq&  \sum_{\gamma \in M_{3}}\left(p^{1+\delta} \vee \sqrt{n}\right)^{-\left(r_{\gamma}-r_{t}\right)} \left(\eta_m^n(\nu)\right)^{-|t|/2} b_n^{(|{\gamma}|-|t|)} p^{(1+4 s)\left(r_{\gamma}-r_{t}\right)} \\ 
    \preceq& \sum_{\gamma \in M_{3}}\left(p^{1+\delta-1-4s} \vee \sqrt{n} p^{-1-4 s}\right)^{-\left(r_{\gamma}-r_{t}\right)} p^{\kappa|t|/2} b_n^{(|{\gamma}|-|t|)} \\ 
    \preceq & \sum_{\gamma \in M_{3}} \left(p^{-\delta / 2} \wedge \frac{p^{1+ \delta/2}}{\sqrt{n}}\right)^{r_{\gamma}-r_t} p^{\kappa|t| / 2} b_n^{(|{\gamma}|-|t|)} \\
    \preceq & \left(p^{-\delta / 2} \wedge \frac{p^{1+ \delta/2}}{\sqrt{n}}\right)^{(J-1) r_{t}+1} p^{\delta(J-1)|t| / 4} \sum_{\gamma \in M_{3}} b_n^{(|{\gamma}|-|t|)} \\
    \preceq &  \rho_n^{(J-1) r_{t}+1} p^{\delta(J-1)|t| / 4}\left(1+b_n\right)^{p} \\
    (\sim & \rho_n^{(J-1)|t|/2}) \longrightarrow 0 
\end{align*}
as $n\rightarrow \infty$. In the above, we used the fact that $\kappa < (J -1)\delta/2$ by condition C4. Note that $\rho_{n}^{r_{\gamma}-r_t} \leq \rho_n^{(J-1) r_{t}+1}$ because $r_{\gamma} > J|t| = J r_t$ and $\rho_n \leq 1$.

\noindent Thus, we have
\begin{equation}\label{eq.large}
    \sum_{\gamma\in M_3} \operatorname{PR}({\gamma}, t)\xrightarrow{P} 0.
\end{equation}

\subsection{Under-fitted Models}\label{sec:proofunderfitted}
First, we will prove that for $c \in (0, 1),$
\begin{equation*}
    P\left[\cup_{\gamma \in M_{4}}\left\{R_{\gamma}-R_{t}<\Delta_{n}(1-c)\right\}\right] \longrightarrow 0,
\end{equation*}
\noindent where $\Delta_n \equiv \Delta_n(J)$ is defined in Condition
C\ref{c.eigen_g}. Since $1_n^{\top}X =0$, by conditions C3 and C4 we
have
\begin{equation*}
    \begin{aligned} 
    R_{\gamma}^{*}-R_{\gamma \vee t} ^{*} &=\left\|\left(P_{\gamma \vee t }-P_{\gamma}\right) Y\right\|^{2} \\
    &=\left\|\left(P_{\gamma \vee t }-P_{\gamma}\right) X_{t} \beta_{t}+\left(P_{\gamma \vee t}-P_{\gamma}\right) X_{t^c} \beta_{t^c}+\left(P_{\gamma \vee t}-P_{\gamma}\right) \epsilon\right\|^{2} \\ 
    &=\left\|\left(P_{\gamma \vee t }-P_{\gamma}\right) X_{t} \beta_{t}+\left(P_{\gamma \vee t }-P_{\gamma}\right) \epsilon\right\|^{2} \\ & \geq\left(\left\|\left(P_{\gamma \vee t }-P_{\gamma}\right) X_{t} \beta_{t}\right\|-\left\|\left(P_{\gamma \vee t }-P_{\gamma}\right) \epsilon\right\|\right)^{2}\\
    &=\left(\left\|\left(I-P_{\gamma}\right) X_{t} \beta_{t}\right\|-\left\|\left(P_{\gamma \vee t }-P_{\gamma}\right) \epsilon\right\|\right)^{2}\\
    &\geq \left(\sqrt{\Delta_n}-\Vert \left(P_{\gamma \vee t }-P_{\gamma}\right)\epsilon \Vert\right)^2,
    \end{aligned}
  \end{equation*}
  for all large $n$. Since $\left\|P_{t} \epsilon\right\|^2/\sigma^2 \sim \chi_{r_t}^2$,  for any $v' \in (0,1)$, we have
    \begin{align}
      \label{eq:m4strgamvt}
    &P\left[\cup_{\gamma \in M_{4}}\left\{R_{\gamma}^{*}-R_{\gamma \vee t }^{*}<\left(1-v'\right)^{2} \Delta_{n}\right\}\right] \nonumber\\ 
    \leq& P\left[\cup_{\gamma \in M_{4}}\left\{\left(\sqrt{\Delta_n}-\left\|\left(P_{\gamma \vee t }-P_{\gamma}\right) \epsilon\right\|\right)^2<(1-v')^2 \Delta_n \right\}\right] \nonumber\\ 
    \leq& P\left[\cup_{\gamma \in M_{4}}\left\{\left\|\left(P_{\gamma \vee t }-P_{\gamma}\right) \epsilon\right\|>v' \sqrt{\Delta_{n}}\right\}\right] \nonumber\\ 
    \leq& P\left[\left\|P_t \epsilon\right\|^2>{v'}^2 \Delta_{n}\right] \nonumber\\ 
    \leq& e^{-c_7 \Delta_{n}},
    \end{align}
  for some constant $c_7 >0$. We also have for any $v'\in (0, 1)$,
\begin{equation*}
    P\left[\cup_{\gamma \in M_{4}}\left\{R_{\gamma \vee t }^{*}-R_{\gamma \vee t }<-\Delta_{n} v' / 2\right\}\right] < e^{-c_8\Delta_{n}},
  \end{equation*}
  for some constant $c_8 >0$. To see this, let
$X_{\gamma \vee t }=U_{n \times r} \Lambda_{r \times r} V_{r \times
  |{\gamma} \vee t |}^{\top}$ be the SVD of $X_{\gamma \vee t }$,
where $r=\text{rank}(X_{\gamma \vee t })$. Then,
$P_{\gamma \vee t }=U U^{\top}$ is the projection matrix onto the
column space of $X_{\gamma \vee t }$ and thus, using $1_n^{\top}X =0$
and equation (4) in the main paper, we have
\begin{align}
  \label{eq:rgamvt}
    R_{\gamma \vee t }^* - R_{\gamma \vee t } 
    &= Y^{\top}(I-U U^{\top})Y - Y^{\top}\left(I + \lambda^{-1}U \Lambda^2 U^{\top}\right)^{-1}Y\nonumber\\
    &= Y^{\top} U\left(\Lambda^2(\lambda I + \Lambda^2)^{-1}-I\right)U^{\top} Y\nonumber\\
    &= \lambda Y^{\top} U(\lambda I + \Lambda^2)^{-1} U^{\top} Y\nonumber\\
    &\leq \left(n\eta_m^n(\nu)/\lambda\right)^{-1} Y^{\top} U U^{\top} Y,
\end{align}
where the last inequality holds because $\lambda I + \Lambda^2 \ge \Lambda^2 \geq n \eta_m^n(\nu) I$.\\

\noindent Since the rank of $U$ is at most $(J+1)|t|$, by \eqref{eq:rgamvt} and \eqref{eq:eigenbound} we have
\begin{align}
   \label{eq:m4rgamvt}
    &P\left[\cup_{\gamma \in M_4} \left\{R_{\gamma \vee t }^* - R_{\gamma \vee t } < -\Delta_n \frac{v'}{2}\right\}\right]\nonumber\\
    \preceq& P\left[\cup_{\gamma \in M_4} \left\{\left(n\lambda^{-1}\eta_m^n(\nu)\right)^{-1} Y^{\top} U U^{\top} Y < -\Delta_n \frac{v'}{2}\right\}\right]\nonumber\\
    \preceq& \exp\left\{-v'n\lambda^{-1}\eta_m^n(\nu) \Delta_n\right\}p^{(J+1)|t|}\nonumber\\
    \preceq& \exp\left\{-p^{2+2\delta}\Delta_n + (J+1)|t|\log{p}\right\}\nonumber\\
    \preceq& e^{-c_8\Delta_n}
\end{align}
\noindent The last inequality above holds because by condition C4
\begin{align*}
    \frac{(J+1)|t|\log{p}}{\Delta_n} \longrightarrow 0 \text{ as } n \rightarrow \infty
\end{align*}
\noindent Then with $R_{\gamma} \geq R_{\gamma}^{*}$, from \eqref{eq:m4strgamvt} and \eqref{eq:m4rgamvt}, we have for any $v \in (0, 1)$,
  \begin{align}
    \label{eq.bd_underfitted}
    &P\left[\cup_{\gamma \in M_{4}}\left\{R_{\gamma}-R_{\gamma \vee t }<\Delta_{n}(1-v)\right\}\right] \nonumber\\ 
    \leq& P\left[\cup_{\gamma \in M_{4}}\left\{R_{\gamma}^{*}-R_{\gamma \vee t }^{*}<\Delta_{n}(1-v / 2)\right]\right.  \nonumber\\ 
    &\ \ \ \ +P\left[\cup_{\gamma \in M_{4}}\left\{R_{\gamma \vee t }^{*}-R_{\gamma \vee t }<-\Delta_{n} v / 2\right\}\right]  \nonumber \\ 
    \leq & 2e^{-c_9 \Delta_{n}} \longrightarrow 0,
    \end{align}
  for some constant $c_9 >0$. Due to \eqref{eq.bd_underfitted} and Lemma
  \ref{lm.tail_prob} with condition C2, for $0<c=3 v<1$, we have
\begin{align*}
    &P\left[\cup_{\gamma \in M_{4}}\left\{R_{\gamma}-R_{t}<\Delta_{n}(1-c)\right\}\right] \\
    \leq& P\left[\cup_{\gamma \in M_{4}}\left\{R_{\gamma}-R_{\gamma \vee t }<\Delta_{n}(1-2 v)\right\}\right]+P\left[\cup_{\gamma \in M_{4}}\left\{R_{\gamma \vee t }-R_{t}<-\Delta_{n} v\right\}\right]\\
    \leq& P\left[\cup_{\gamma \in M_{4}}\left\{R_{\gamma}-R_{\gamma \vee t }<\Delta_{n}(1-2v)\right\}\right]+P\left[\cup_{\gamma \in M_{4}}\left\{R_{t}-R_{\gamma \vee t }>\Delta_{n} v^{2}\right\}\right] \\
    \leq& \exp \left\{-c_9 \Delta_{n}\right\}+P\left[R_{t}-R_{t}^{*}>\Delta_{n} v^{2}/2 \right]+P\left[\cup_{\gamma \in M_{4}} \left\{R_{t}^{*}-R_{\gamma \vee t }^{*}\right\}>\Delta_{n} v^{2}/2\right] \\  
    \leq& \exp \left\{-c_9 \Delta_{n}\right\} + \exp \left\{-c' \Delta_{n}\right\} + P\left[\chi_{J|t|}^{2}>\Delta_{n} v^{2}/2\right] \\
    \leq& 3\exp \left\{-c_{10} \Delta_{n}\right\} \rightarrow 0, 
\end{align*}
for some constant $c_{10} >0$.
\noindent Therefore, restricting to the high probability event 
$$\left\{R_{\gamma}-R_{t} \geq \Delta_{n}(1-c), \forall {\gamma} \in M_{4}\right\}$$, by corollary~\ref{cl.post_ratio} and \eqref{eq:eigenbound} we get
  \begin{align}
    \label{eq.ratiobd_underfiited1}
    \sum_{\gamma \in M_{4}} P R({\gamma}, t)
    \preceq& \sum_{\gamma \in M_{4}} \left(n\eta_m^n(\nu)/\lambda\right)^{-(r_{\gamma}^{*}-r_t)/2}\left(\eta_m^n(\nu)\right)^{-|\gamma^c \wedge t|/2} b_n^{(|{\gamma}|-|t|)} \exp \left\{-\frac{1}{2 \sigma^{2}}\left(R_{\gamma}-R_{t}\right)\right\}\nonumber\\
    \preceq& \sum_{\gamma \in M_{4}}\left(p^{2+3 \delta} \vee n\right)^{|t| / 2} p^{\delta|t|/2} b_n^{|{\gamma}|-|t|} \exp \left\{-\Delta_{n}(1-c) / 2 \sigma^{2}\right\},
    \end{align}
because $r_t-r_{\gamma}^{*} < r_t =|t|$ and $\eta_m^n(\nu)=(n\eta_m^n(\nu)/\lambda)/(n/\lambda) \succeq (p^{2+2 \delta} \vee n) / (p^{2+3 \delta} \vee n)=p^{-\delta}$ due to condition C\ref{c.shrink_rates} and condition C\ref{c.eigen_g}. Then by \eqref{eq.ratiobd_underfiited1} we have
\begin{equation}\label{eq.ratiobd_underfiited2}
    \begin{aligned}
      \sum_{\gamma \in M_{4}} P R({\gamma}, t) &\preceq \exp \left\{-\frac{1}{2 \sigma^{2}}\left(\Delta_{n}(1-c)-\sigma^{2}|t| \log \left(p^{2+3 \delta} \vee n\right)-\sigma^{2}(2+\delta)|t|\log p\right)\right\} \sum_{\gamma \in M_{4}} b_n^{|{\gamma}|}\\
      &\preceq \exp \left\{-\frac{1}{2 \sigma^{2}}\left[\Delta_{n}(1-c)-\sigma^{2}|t|( \log (p^{4+4 \delta} \vee n p^{2+\delta}))\right]\right\} (1+ b_n)^p\\
    &\preceq \exp \left\{-\frac{1}{2 \sigma^{2}}\left(\Delta_{n}(1-c)-c_{11} \tau_{n}\right)\right\} \\
    &\longrightarrow 0 \text{ as } n \rightarrow 0,
    \end{aligned}
\end{equation}
\noindent where $c_{11} >0$ and $\tau_n=5(1+\delta)\log(\sqrt{n} \vee p).$
To see the last inequality, we consider two cases. First, if
$\sqrt{n} < p, \tau_n = \log(p^{5+5\delta})$ and
$n p^{2+\delta} < p^{4+\delta} < p^{4+4\delta}$, and thus
$\log(p^{4+4 \delta} \vee n p^{2+\delta}) = \log(p^{4+4\delta}) <
\log(p^{5+5\delta}) = \tau_n.$ Then, if $\sqrt{n} > p,$
$\tau_n = \log(n^{5(1+\delta)/2})$. Then,
$p^{4+4\delta} < p^{5+5\delta} < n^{5(1+\delta)/2}$ and
$n p^{2+\delta} < n^{2+\delta/2} < n^{5(1+\delta)/2}.$ Therefore,
$\log(p^{4+4 \delta} \vee n p^{2+\delta}) < \log(n^{5(1+\delta)/2}) =
\tau_n.$ Also, the last line of \eqref{eq.ratiobd_underfiited2} holds
because by condition C\ref{c.eigen_g},
$\Delta_n \succ \log(\sqrt{n} \vee p)$, that is, $\Delta_n > \tau_n$.
\noindent Hence, we have 
\begin{equation}\label{eq.underfitted}
    \sum_{\gamma\in M_4} \operatorname{PR}({\gamma}, t)\xrightarrow{P} 0.
\end{equation}
Now, combining \eqref{eq.verylarge}, \eqref{eq.overfitted}, \eqref{eq.large} and \eqref{eq.underfitted} we get $\sum_{\gamma \not = t} \operatorname{PR}({\gamma}, t)\xrightarrow{P} 0,$
which proves Theorem \ref{thm.consistency}.

\section{Proof of Theorem \ref{thm.consistency.nosigma}}\label{sec:proofUnknownSigma}
Next, we will show that with a prior on $\sigma^2$ in \eqref{subeq:priorInterceptVariance}, the model selection consistency holds under the assumption that $P(\gamma \in \widetilde{M}) = 0.$ Note that since $\log p = o(n)$ and $\nu' > \nu$, we have $M_1 \subset \widetilde{M}$ eventually. Thus $P(\gamma \in M_1) = 0$ for all large $n$. 
Therefore, we shall show that
$\sum_{\gamma\in \widetilde{M}_k} \widetilde{\operatorname{PR}}({\gamma},
t)\xrightarrow{P} 0$ for $k = 2,3,4$ where $\widetilde{M}_k = M_k\cap \widetilde{M}$ and $\widetilde{\operatorname{PR}}({\gamma},
t) \equiv P(\gamma | Y)/P(t | Y)$. By (2d) and (3) of the main paper, we have
\[
  f(\gamma |Y) = c_{n,p} Q_{\gamma}  b_n^{|\gamma|} (1- w)^p R_{\gamma}^{-(n-1)/2}.
\]
By condition C2 and Lemma
\ref{lm.Qgamma}, we then get
\begin{equation}
  \label{eq:prtil}
    \widetilde{\operatorname{PR}}({\gamma}, t)
     \preceq \left(n\eta_m^n(\nu)/\lambda\right)^{-(r_{\gamma}^{*}-r_t)/2} \left(\eta_m^n(\nu)\right)^{-|t \wedge {\gamma}^{c}| / 2} b_n^{(|\gamma|-|t|)} (R_{\gamma}/R_{t})^{-(n-1)/2}.
\end{equation}
\noindent Define
\begin{equation*}
    \zeta_n := \frac{R_t}{n\sigma^2}-1.
\end{equation*}

\noindent Due to our Lemma \ref{lm.tail_prob} and Lemma A.2(ii) of \citet{nari:he:2014}, for $\phi > 0$, we have
\begin{equation}
    \begin{aligned}
        P(|\zeta_n|>2\phi) &= P\left(\left| \frac{R^{*}_t}{n\sigma^2}-1 + \frac{R_t-R^{*}_t}{n \sigma^2} \right|>2\phi\right) \\
        &\leq P\left(\left| \frac{R^{*}_t}{n\sigma^2}-1 \right|>\phi\right) + P\left(R_t-R^{*}_t \geq \phi n \sigma^2\right) \\
        &\leq 2\exp(-c_{12} n),
    \end{aligned}
\end{equation}
for some positive quantity $c_{12}$ depending on $\phi$. From \eqref{eq:prtil} we have 
\begin{align}\label{eq:prnosigma} 
    \widetilde{\operatorname{PR}}({\gamma}, t)
    & \preceq \left(n\eta_m^n(\nu)/\lambda\right)^{-(r_{\gamma}^{*}-r_t)/2} \left(\eta_m^n(\nu)\right)^{-|t \wedge {\gamma}^{c}| / 2} b_n^{(|\gamma|-|t|)} \left(1 + \frac{R_{\gamma}-R_t}{n\sigma^2(1+\zeta_n)}\right)^{-\frac{n-1}{2}}.
\end{align}
Define $z_n := (r_\gamma-r_t)\log p /n.$ Note that for models in $M_2$, $r_{\gamma} > r_t$. Since condition C6 is in force, we have $ z_n< 1/(2+\nu')$, and choose $s > 0$ and $\Tilde{\phi} > 0$ such that $2(1+4s)/\{(1-\Tilde{\phi})(2+\nu')\} < 1$ and 
\[1 < \dfrac{(1+4s)}{(1-\Tilde{\phi})/\left\{1 - 2(1+4s)/[(1-\Tilde{\phi})(2+\nu')]\right\}} < (\delta + 1)/2,\]
which is possible since $\nu'\delta > 2.$ Consequently,
\begin{equation}\label{eq:xnbound}
    x_n := -\log\left(1-2\frac{1+4s}{1-\Tilde{\phi}}~z_n\right) < \frac{2(1+4s)z_n}{(1-\Tilde{\phi})\{1-2(1+4s)z_n/(1-\Tilde{\phi})\}} < 2(\delta/2+1)z_n.
  \end{equation}
  where the first inequality follows from the fact that $-\log(1-x) < x/(1-x)$ for $0 < x  < 1.$ Using the similar way as in Section \ref{sec:proofoverfitted}, we only
  consider the high probability event
  $\left\{\cap_{\left\{d>r_{t}\right\}} U(d)^c\right\}
  \cap\left\{\left|\zeta_{n}\right|<\tilde{\phi}\right\}$, where
  $U(d)$ is defined the same as in Section
  \ref{sec:proofoverfitted}. Note that on $U(d)^c$,
  $1 + (R_{\gamma} - R_t)/(n\sigma^2) > 1-2(1+4s)z_n$. Note that on
  $M_2$, $\gamma^c \wedge t$ is empty. Then, due to
  \eqref{eq:prnosigma}, \eqref{eq:xnbound}, and \eqref{eq:eigenbound}
  we obtain

\begin{align*}
     \sum_{\gamma \in M_{2}} \widetilde{\operatorname{PR}}({\gamma}, t) & \preceq
     \sum_{\gamma \in M_{2}} \left(n\eta_m^n(\nu)/\lambda\right)^{-\left(r_{\gamma}^{*}-r_{t}\right)/2} b_n^{(|{\gamma}|-|t|)} \exp \left\{\left(\frac{n-1}{2}\right) x_{n}\right\}\\
     & \preceq \sum_{\gamma \in M_{2}}\left(p^{1+\delta} \vee \sqrt{n}\right)^{-\left(r_{\gamma}-r_{t}\right)} b_n^{(|\gamma|-|t|)} \exp \left\{\left(\frac{n}{2}\right) x_{n}\right\} \\ 
     & \preceq \sum_{\gamma \in M_{2}}\left(p^{1+\delta} \vee \sqrt{n}\right)^{-\left(r_{\gamma}-r_{t}\right)} b_n^{(|\gamma|-|t|)} p^{(\delta/2+1)\left(r_{\gamma}-r_{t}\right)} \\ 
     & \sim \rho_n \rightarrow 0, \text { as } n \rightarrow \infty,
\end{align*}
where $\rho_n$ is defined in Section \ref{sec:proofoverfitted}.
Also, following from the proof for large models in Section \ref{sec:prooflarge}, we can show that
\begin{equation*}
    \sum_{\gamma \in M_{2} \cup M_{3}}  \widetilde{\operatorname{PR}}({\gamma}, t) \stackrel{\mathrm{P}}{\longrightarrow} 0.
\end{equation*}
For under-fitted models in $M_4$, if $\Delta_n=o(n)$, similar to
\eqref{eq.ratiobd_underfiited1} and \eqref{eq.ratiobd_underfiited2}
restricting to the high probability event
$\left\{R_{\gamma}-R_{t} \geq \Delta_{n}(1-c)\right\} \cap\left\{\left|\zeta_{n}\right|<\tilde{\phi}\right\}$, we get
\begin{align*} 
\sum_{\gamma \in M_{4}} \widetilde{\operatorname{PR}}({\gamma}, t) & \preceq \sum_{\gamma \in M_{4}}\left(n\eta_{m}^{n}\right/\lambda )^{|t| / 2}\left(\eta_m^n(\nu)\right)^{-|\gamma^c \wedge t| / 2} b_n^{|{\gamma}|-|t|} \left(1+\frac{R_{\gamma}-R_t}{n\sigma^2(1+\zeta_n)}\right)^{-\frac{n}{2}}\\
& \preceq \sum_{\gamma \in M_{4}}\left(p^{2+3 \delta} \vee n\right)^{|t| / 2} p^{\delta|t| / 2} b_n^{(|\gamma|-|t|)}  \exp \left\{-\frac{\Delta_{n}(1-c)}{2 \sigma^{2}(1+\tilde{\phi})}\right\}\\ 
& \preceq  \exp \left\{-\frac{1}{2 \sigma^{2}}\left(\Delta_{n}(1-c)/(1+\tilde\phi)-\sigma^{2}|t| \log \left(p^{2+3 \delta} \vee n\right)-\sigma^{2}|t|(2+\delta) \log p\right)\right\} \\
& \preceq \exp \left\{-\frac{1}{2 \sigma^{2}}\left(\Delta_{n}\left(1-c'\right)-\tau_n\right)\right\} \\ & \rightarrow 0, \text { as } n \rightarrow \infty.
\end{align*}
If $\Delta_n \sim n$, then by taking $\tilde{\phi}<1/2$, we have for some $v'>0$ and $c'>0$
\begin{align*}
   \sum_{\gamma \in M_{4}} \widetilde{\operatorname{PR}}({\gamma}, t)
   & \preceq\left(p^{2+3 \delta} \vee n\right)^{|t| / 2} p^{(1+\delta/2)|t|}\left(1+\frac{\Delta_{n}\left(1-c'\right)}{4 n \sigma^{2}}\right)^{-\left(\frac{n}{2}\right)} \\ 
   & \preceq\left(p \vee n\right)^{(2+3 \delta)|t|} e^{-v' n} \rightarrow 0, \text { as } n \rightarrow \infty.
\end{align*}

\end{document}